\documentclass[12pt]{article}
\usepackage{setspace} 
\usepackage[titletoc,title]{appendix}
\usepackage[dvips]{graphicx} 
\usepackage{amsmath,amsthm,amsfonts}
\usepackage{epsfig}
\usepackage{subfig}
\usepackage{hyperref}
\hypersetup{
colorlinks = true,
citecolor=blue,
urlcolor=blue,
linkcolor=blue,
pdfauthor = {Daniel Hackmann},
pdfkeywords = { },
pdftitle = {},
pdfpagemode = UseNone
}
\usepackage{mathtools}
\usepackage{diagbox}
\usepackage{setspace}
\usepackage{amsthm}
\usepackage{subfig}
\usepackage{appendix}
\usepackage{enumerate}
\usepackage{changepage}
\usepackage{calc}
\graphicspath{{./}}
\usepackage{calligra} 
\DeclareMathAlphabet{\mathcalligra}{T1}{calligra}{m}{n} 
\DeclareFontShape{T1}{calligra}{m}{n}{<->s*[2.2]callig15}{}
\numberwithin{equation}{section} 
    \def\qed{\hfill$\sqcap\kern-8.0pt\hbox{$\sqcup$}$\\}
    \def\beq{\begin{eqnarray}}
    \def\eeq{\end{eqnarray}}
    \def\beqq{\begin{eqnarray*}}
    \def\eeqq{\end{eqnarray*}}
    
    \def\re{\textnormal {Re}}

    \def\p{{\mathbb P}}
    \def\e{{\mathbb E}}
    
    \def\r{{\mathbb R}}

    \def\c{{\mathbb C}}

    \def\d{{\textnormal d}}

    \def\ind{{\mathbb I}}

    \def\eq{\textbf{e}(q)}

    	\def\hr{\hat{\rho}}
    	
    	\def\hz{\hat{\zeta}}

    	\def\hN{\hat{N}}
    	\def\hM{\hat{M}}
	\newtheorem{theorem}{Theorem}
	\newtheorem{lemma}{Lemma}
	\newtheorem{proposition}{Proposition}
	\newtheorem{corollary}{Corollary}
	\theoremstyle{definition}
	\newtheorem{example}{Example}

	\newtheorem{remark}{Remark}

	\newcommand\xqed[1]{%
  	\leavevmode\unskip\penalty9999 \hbox{}\nobreak\hfill
  	\quad\hbox{#1}}
	\newcommand\demo{\xqed{$\dagger$}}
	\newcommand\rqed[1]{%
  	\leavevmode\unskip\penalty9999 \hbox{}\nobreak\hfill
  	\quad\hbox{#1}}
	\newcommand\rdemo{\rqed{$\ddagger$}}

\title{Analytic techniques for option pricing under a hyperexponential L\'{e}vy model}
\author{
{Daniel Hackmann
\footnote{  
E-mail: dan@danhackmann.com. Web: www.danhackmann.com.}}
 }
 
 \date{\today}

\begin{document}

\maketitle
\abstract{\noindent We develop series expansions in powers of $q^{-1}$ and $q^{-1/2}$ of solutions of the equation $\psi(z) = q$, where $\psi(z)$ is the Laplace exponent of a hyperexponential L\'{e}vy process. As a direct consequence we derive analytic expressions for the prices of European call and put options and their Greeks (Theta, Delta, and Gamma) and a full asymptotic expansion of the short-time Black-Scholes at-the-money implied volatility. Further we demonstrate how the speed of numerical algorithms for pricing exotic options, which are based on the Laplace transform, may be increased.
\section{Introduction}
A hyperexponential L\'{e}vy process $X$ is one with a L\'{e}vy measure of the form
\begin{align*}
\nu(\d x) = \ind(x < 0)\sum_{\ell=1}^{\hat{N}}\hat{a}_{\ell}\hat{\rho}_{\ell}e^{\hat{\rho}_{\ell}x} \d x + \ind(x > 0)\sum_{\ell=1}^{N}a_{\ell}\rho_{\ell}e^{-\rho_{\ell}x} \d x,
\end{align*}
where the $\{\hat{a}_{\ell}\}_{1\leq\ell\leq \hat{N}}$ and $\{a_{\ell}\}_{1\leq\ell\leq N}$ are all positive real numbers and $0 < \rho_1 < \rho_2 < \ldots < \rho_{N-1} < \rho_{N}$ and $0 < \hat{\rho}_1 < \hat{\rho}_2 < \ldots < \hat{\rho}_{N-1} < \hat{\rho}_{N}$ hold. The Laplace exponent $\psi(z) := \frac{1}{t}\log\left(\e\left[e^{zX_t}\right]\right)$ has the form
\begin{align}\label{eq:main}
\psi(z) = \frac{\sigma^2z^2}{2} + az + z\sum_{\ell=1}^{N}\frac{a_{\ell}}{\rho_{\ell} - z} - z\sum_{\ell=1}^{\hat{N}}\frac{\hat{a}_{\ell}}{\hat{\rho}_{\ell} + z},\quad -\hat{\rho}_1 < \re(z) < \rho_1,
\end{align}
where $a \in \r$ and $\sigma \geq 0$. When $\sigma > 0$ hyperexponential processes are also called hyperexponential diffusions or hyperexponential jump diffusions in the literature. \\ \\
\noindent Despite their apparent simplicity -- they are compound Poisson processes plus a Brownian motion component when $\sigma > 0$ -- they have been studied extensively in the literature for a number of reasons. First, hyperexponential processes are dense in the $\mathcal{CM}$-class  of processes, i.e. those L\'{e}vy processes with completely monotone jump densities (also known as generalized hyperexponential processes)  \cite{Jeannin}. The $\mathcal{CM}$-class includes infinite activity models like the Variance Gamma (VG) process, the Normal Inverse Gaussian (NIG) process, and the CGMY/Kobol/Generalized Tempered Stable process, which have become very popular in finance. Second, there are a number of fast and accurate algorithms that exploit this first quality, i.e. methods by which a $\mathcal{CM}$-class process can be approximated by a hyperexponential process arbitrarily well \cite{crosby_lesaux,al_and_me}. Third, because $\psi(z)$ can be extended to a rational function with real poles on $\c$, hyperexponential processes are ``analytically tractable". For example, we have analytic expressions for the Laplace transform (in $t$) of the distribution of $X_t$ (see Theorem \ref{theo:atq}) and the Wiener-Hopf factors \cite{exp_phase}. For financial applications,  under the assumption of an exponential model for the stock price, analytic expressions for the Laplace transform of the prices of barrier \cite{Jeannin,Cai_double,sepp} and look-back \cite{Cai_Kou} options, for the double Laplace transform of the price of an Asian option \cite{Cai_Kou_Asian_options}, and for the prices of Russian options and for perpetual American strangles \cite{exp_phase,boy_strangle} are known. If we restrict $N=\hat{N} = 1$ to get the so-called double exponential or Kou model, we have also analytic expressions for prices of European call and put options and European options on futures contracts \cite{kou_sad}, as well as perpetual American options \cite{kou_wang}.\\ \\
In almost all of the cases mentioned above, the formula for the derivative price, or the Laplace transformed price, is expressed in terms of the solutions of the equation 
\begin{align}\label{eq:main_idea}
\psi(z) = q,\quad q > 0.
\end{align} 
If we exclude those cases where there are fewer than four solutions, then the solutions need to be determined numerically. As a practical matter, finding solutions to \eqref{eq:main_idea} is a time consuming part of the algorithm for inverting the Laplace transform to obtain option prices (Asian options, barrier options, look-back options), especially because it becomes necessary to solve \eqref{eq:main_idea} for $q \in \c$.\\ \\
\noindent The main idea behind this article straightforward: we develop convergent series in powers of $q^{-1}$ (when $\sigma = 0$) and $q^{-1/2}$ (when $\sigma > 0$) of the solutions of \eqref{eq:main_idea} for $q\in\c$ with $\vert q \vert$ large enough. Since the series converge quite rapidly, an immediate consequence is that the (truncated) series may be used to speed up algorithms for determining derivative prices based on numerical inversion of the Laplace transform.\\ \\
While this is a useful result, further, interesting results follow from the main idea. We are also able to use the expansions to develop analytic expressions for the prices European call and put options and their Greeks. This is rather rare in exponential L\'{e}vy models, to the best of the author's knowledge there are only two other L\'{e}vy processes for which this is true: a) Merton's model \cite{merty} and b) Kou's model \cite{kou_sad}. The resulting expressions involve series of functions in $T$, the time of expiry of the option, which when $\sigma = 0$ are, in fact, just Taylor series. In the at-the-money (ATM) case, when $\sigma > 0$, the formulas are essentially series in powers of $T^{1/2}$; this allows us to develop a full asymptotic expansion of the short-time ATM Black-Scholes implied volatility. Implied volatiles, together with short-time asymptotic expansions of call option prices, have seen a large amount of recent interest in the financial mathematics literature owing to their application to the calibration problem (see for example \cite{figruot} and the references therein).\\ \\
It should be noted that we are generalizing Kou's results \cite{kou_sad}. While Kou also develops analytic formulas for European call and put option prices, his approach relies on the decomposition of sums of double exponential random variables; this technique does not seem to have a natural extension to the general case, where the number of exponential factors in the L\'{e}vy density exceeds two.\\ \\
\noindent  Our approach is therefore rather different and analytical in nature, relying on results from complex analysis and the theory of Laplace transforms. We devote Section \ref{sect:tools} of the article to reviewing the relevant theory and developing notation. In Section \ref{sect:key_hyp} we gather some key results for hyperexponential processes and develop the series expansions of the solutions of \eqref{eq:main_idea}. Then in Section \ref{sect:option} we develop analytic formulas for European option prices and Greeks, derive a full asymptotic expansion of the ATM implied volatility, and consider several numerical examples. In one of these, we show how the speed of computing the price of a digital barrier option via inverting the Laplace transform can be at least doubled. In another, we demonstrate that our formulas for put and call option prices are much faster for computing short to medium-time prices than the traditional approach based on numerical Laplace inversion (faster by a factor of at least five for 100 option prices with maturities shorter than 0.5). Software used to compute the various examples given throughout the article can be obtained from the author's website.
\section{Tools from complex analysis}\label{sect:tools}
\subsection{Basic notation}
Assuming $R>0$ and $z_0 \in \r$ we define
\begin{align*}
\c^+ := \{z \in \c : z \notin (-\infty,0] \},\quad \c_R := \{z \in \c : \vert z \vert > R\},\quad\text{ and }\quad\mathbb{H}_{z_0} := \{z \in \c: \re(z) > z_0 \}, 
\end{align*}
and using these $\c_R^{+} := \c^+ \cap \c_R$ and $\mathbb{H} := \mathbb{H}_0$. The notation $\mathbb{Z}^+$ refers to the non-negative integers, with the analogous meaning for the notation $\mathbb{Z}^-$. We will use $B$ to denote an open ball in $\c$ centered at $0$, and $B_0$ to denote a punctured open ball excluding the point 0. If we want to be specific about the radius $R$ we will write $B(R)$ and $B_0(R)$. \\ \\
The collection of solutions $w$ of the equation $w^k = z$, $k\in\mathbb{N}$ are denoted ${z}^{1/k}_m$. It follows that ${z}^{1/k}_m$ is a multi-valued function (see pg. 24 in \cite{mark} for a rigorous definition) taking exactly $k$ values for all $z \neq 0$. The principal branch of ${z}^{1/k}_m$ will be denoted simply $z^{1/k}$. As usual, the principal branch is that branch for which ${1}^{1/k} = 1$. Further, we define ${z}^{n/k}_m := ({z}^{1/k}_m)^n$ for $n \in \mathbb{Z}$, which is again a $k$-valued function when $k$ is relatively prime to $n$. Our primary concern will be the case $k=2$. In this scenario, the non-principal branch can be expressed in terms of the principal branch as $-z^{1/2}$; the two branches of ${z}^{n/2}_m$ are then just given by $z^{n/2} := (z^{1/2})^n$ and $(-1)^nz^{n/2}$. The notation $\log(z)$ always refers to the principal branch of the logarithm, i.e. that branch for which $\log(1) = 0$. The notation $\Gamma(z)$ refers to the gamma function.
\subsection{Working with series}\label{sect:working}
We will work with a Laurent series
\begin{align}\label{eq:basic_taylor}
f(z) := \sum_{n=k}^{\infty}f_nz^n,
\end{align}
where $\{f_n\}_{n\geq k} \subset \c$ and $k \in \mathbb{Z}$. We assume that $f_k \neq 0$, that the series converges on $B_0(R)$, and that $f(z) \in \r$ for $z \in \r$. The notation $_{k}\bar{f}_n$ denotes the $(n+k+1)$-tuple
\begin{align}\label{def:-kfn}
_{k}\bar{f}_n := (f_{k},\,f_{k+1},\,\ldots,\,f_n) \in \c^{n+k+1},\quad n\geq k,
\end{align}
and we define $\bar{f}_n :=\, _0\bar{f}_n$. Note that we can apply this latter notation to any sequence $\{a_n\}_{n\geq k}$, not necessarily only in the context of an underlying series.  \\ \\
\noindent If $f(z)$ converges on some $B_0(R)$, then it is well known (see Theorems 16.1 and 16.2 in \cite{mark}), that $1/f(z)$ also has a convergent series representation on $B_0(R')$ for some $R' > 0$. Further, the $n$-th coefficient is a function $r_n:\c^{n-k+1} \rightarrow \c$ that depends only on $_{k}\bar{f}_{n-2k}$, which can be easily computed (see Theorems 1.3 and 2.3d/f in \cite{cca}). We have
\begin{align}\label{eq:recip}
\frac{1}{f(z)} = \sum_{n=-k}^{\infty}r_nz^{n},\quad\text{where}\quad r_{-k} := r_{-k}(_{k}\bar{f}_{k}) := \frac{1}{f_{k}},\quad r_{-k+1} := r_{-k+1}(_{k}\bar{f}_{k+1}):= -\frac{f_{k+1}}{(f_{k})^{2}},
\end{align}
and for $n \geq -k+2$,
\begin{align}\label{def:rn}
r_{n} := r(_{k}\bar{f}_{n+2k}) := \frac{(-1)^{n+k}}{(f_{k})^{n+k+1}}
\begin{vmatrix}
f_{k+1} & f_{k+2} & \ldots & f_{n+2k} \\
f_{k} & f_{k+1} & \ldots & f_{n+2k-1} \\
0 & f_{k} & \ldots & f_{n+2k-2} \\
& & \ldots & \\
0 & 0 & f_{k} & f_{k+1}
\end{vmatrix}.
\end{align} 
Similarly, assuming that $k \geq 0$ and that $f(z)$ converges on some ball $B$, then for $c>0$ and $z \in B$ we have
\begin{align}\label{eq:power}
c^{f(z)} = \sum_{n=0}^{\infty}p_nz^{n},\quad p_0 := p_0(c,f_0):= c^{f_0},
\end{align}
and for $n \in \mathbb{N}$,
\begin{align}\label{def:pn}
p_n :=p_n(c,\bar{f}_n) := \frac{1}{n!}\sum_{m=1}^{n}c^{f_0}\left(\log(c)\right)^m B_{n,m}\left(1! f_1,\,2!f_2,\,\ldots,\,(n-m+1)!f_{n-m+1}\right),
\end{align}
where $\{B_{n,m}\}_{n\geq0,m\geq0}$ are the exponential partial Bell partition polynomials (see Definition 11.2 in \cite{char}). The derivation of \eqref{def:pn} follows from Fa{\`a} di Bruno's generalization of the chain rule for higher derivatives (see Lemma 1.3.1 in \cite{realan}). Note that we will write $p_n(c,-\bar{f}_n)$ for the coefficients of the series expansion of $c^{-f(z)}$.\\ \\
\noindent If $F = \{f_1(z),\,f_2(z),\,\ldots f_N(z)\}$ is a collection of series of the form \eqref{eq:basic_taylor}, we will write the $n$-th coefficient of the $i$-th series as $f_{i,n}$, and the index of the first non-zero coefficient of the $i$-th series as $k_i$. Then if $k_i \geq 0$ for all $f_i(z) \in F$, and all members of $F$ converge on a common ball $B$, then it is well known that for $z \in B$
\begin{align}\label{def:sn}
\sum_{i=1}^Nf_i(z) = \sum_{n=j}^{\infty}s_nz^{n},\quad s_n := s_n(f_{1,n},\,f_{2,n},\,\ldots,\,f_{N,n}) := \sum_{i=1}^{N}f_{i,n},
\end{align}
where $j = \min\{k_1,\,k_1,\,\ldots,\,k_N\}$ and we set $f_{i,n} = 0$ whenever $n < k_i$. Similarly, for $z \in B$,
\begin{align}\label{def:mn}
\prod_{i=1}^Nf_i(z) = \sum_{n=j}^{\infty}m_nz^{n},\quad m_n := m_n(\bar{f}_{1,n},\,\bar{f}_{2,n},\,\ldots,\, \bar{f}_{N,n}):= \pi_n^{N},
\end{align}
where $j = \sum_{i=1}^Nk_i$, and $\pi_n^{N}$ is defined recursively with $\pi_n^{1} : = f_{1,n}$ and
\begin{align*}
\pi_n^{i} := \sum_{k=0}^{n}\pi_{k}^{i-1}f_{i,n-k},\quad i \in \{2,\,3,\,\ldots,\,N\};
\end{align*}
again we abide by the convention $f_{i,n} = 0$ whenever $n < k_i$.\\ \\
We can also consider objects of the form $f(t(z))$, for some map $t:\c \rightarrow \c$. It is clear, that if for some $z_0$ we have $f(t(z_0)) \in \c$, then $f(z)$ is absolutely convergent on $B_0(\vert t(z_0) \vert)$ and therefore that $f(t(z))$ converges absolutely on $t^{-1}[B_0(\vert t(z_0) \vert)]$. Series expressions for $1/f(t(z))$ and $c^{f(t(z))}$ in powers of $t(z)$ can then be derived simply by replacing $z$ in \eqref{eq:recip} and \eqref{eq:power} by $t(z)$, with the understanding that the series converge on sets the form $t^{-1}[B_0]$; we will avoid any cases where these sets are empty. Analogously, formulas \eqref{def:sn} and \eqref{def:mn} also hold if we replace $z$ by $t(z)$, provided we consider the proper domain.\\ \\
Such series, i.e. those where $z$ is replaced by some transformation of $z$, occur naturally when we wish to derive the inverse series. This can be done via the Lagrange Inversion Theorem (Theorems 3.4 and 3.6 in  \cite{mark2}). This tells us that if $k \geq 1$ and $w = f(z)$ on $B(R)$, then there exists $R' > 0$ such $f(z)$ has a $k$-valued inverse $f^{-1}(w)$ on $B(R')$ of the form
\begin{align}\label{eq:invlagm}
f^{-1}(w) = \sum_{n=1}^{\infty}l_n w_m^{n/k},
\end{align}
where
\begin{gather}\label{eq:invlagm2}
l_n = \frac{1}{n!}\left(\frac{\d^{n-1}}{\d z^{n-1}}\left(\kappa(z)\right)^{n} \right)_{z=0},\quad
\kappa(z) := \frac{z}{\left( f(z) \right)_s^{1/k}},
\end{gather}
and where $(f(z))_s^{1/k}$ is any single-valued branch of the multiple-valued $(f(z))_m^{1/k}$. An explicit formula for $l_n$ can be obtained by choosing
\begin{align}\label{eq:new_form_kappa}
(\kappa(z))^n = f_k^{-n/k}\left(1 + \frac{f_{k+1}}{f_k}z + \frac{f_{k+2}}{f_k}z^2 + \ldots\right)^{-n/k}.
\end{align}
and applying Fa{\`a} di Bruno's Formula. This yields
\begin{align}\label{def:ln}
l_n := l_n(_k\bar{f}_{n+k-1}) := \frac{1}{n!f_k^{n/k}}\sum_{m=1}^{n-1}(-1)^m\left(\frac{n}{k}\right)_mB_{n-1,m}\left(1!\frac{f_{k+1}}{f_k},\,2!\frac{f_{k+2}}{f_k},\,\cdots,\,(n-m)!\frac{f_{k+n-m}}{f_k}\right),
\end{align}
where $(x)_i := x(x+1)\cdots(x + i -1)$ denotes the rising factorial. Formula \ref{def:ln} is valid for all $n \geq 2$; for $n=1$ we set $l_1 := l_1(_k\bar{f}_k) :=  1/\sqrt[k]{f_{k}} \neq 0$, which can be justified via \eqref{eq:new_form_kappa} and \eqref{eq:invlagm2}.

\begin{remark}
While the formulas in this section are a little daunting, it should be noted that most software packages that have a symbolic computation component have routines to handle series manipulations, even for fractional powers of the argument. Therefore, it is unnecessary to carry out computations by hand, or even write computer programs to compute, for example, the coefficients of a reciprocal or inverse series. For the remainder of the paper we use a Mathematica implementation to perform all series manipulations; the corresponding software can be found on the author's web page. All computations are carried out on a machine with 32GB of memory and an Intel i7-2600K CPU @ 3.40GHz.\rdemo
\end{remark}
\subsection{Termwise inversion of Laplace transforms represented by series}
In this brief section we recall an important result of Doetsch \cite{doetsch} concerning the inversion of Laplace transforms given by series and state a useful corollary. These will be the keys to developing series expansions of option prices in Section \ref{sect:option}. Here and throughout we use the notation $\mathcal{L}\{f\}(z)$ or $\mathcal{L}\{f(t)\}(z)$  to denote the Laplace transform of the function $f(t)$, which is defined
\begin{align}\label{eq:laplace}
\mathcal{L}\{f\}(z) := \int_0^{\infty}e^{-zt}f(t)\d t.
\end{align}
A key result for Laplace transforms is that if the integral converges for some $z_0 \in \r$ then it converges for all $z \in \mathbb{H}_{z_0}$ and is an analytic function of $z$ there. Note that in this section, the notation $f_n$ is used to denote functions rather than constant coefficients as was the case in the last section.
\begin{theorem}[Satz 30.1 in \cite{doetsch}]\label{theo:tbt}
Suppose that, for some collection of functions $\{f_n(t)\}_{n\geq0}$, the Laplace transforms $G_n(z) := \int_0^{\infty} e^{-zt}\vert f_n(t)\vert\d t$ and $F_n(z) := \mathcal{L}\{f_n\}(z)$ exist for every $n \in \mathbb{Z}^+$ on some common half-plane $\mathbb{H}_{z_0}$. Further, suppose that the series
\begin{align}\label{eq:doetsch}
G(z) := \sum_{n=0}^{\infty}G_n(z),\quad\text{ and therefore also}\quad F(z) := \sum_{n=0}^{\infty}F_n(z),
\end{align} 
converge on $\mathbb{H}_{z_0}$. Then, $\sum_{n=0}^{\infty}f_n(t)$ converges absolutely and for almost all $t \geq 0$ to a function $f(t)$ . Further $\mathcal{L}\{f(t)\}(z) = F(z)$.
\end{theorem}
\noindent The following Corollary follows directly from Theorem 6 together with Satz 5.1, 5.5 and Satz 30.2 in \cite{doetsch}. In particular Satz 30.2 is a generalized version of the following.
\begin{corollary}\label{cor:main}
Suppose $f(t)$ is a continuous function on $[0,\infty)$ such that for some $0 \leq R < \infty$
\begin{align*}
\mathcal{L}\{f\}(z) = \sum_{n=k+j}^{\infty}\frac{a_n}{z^{n/k}}, \quad z \in \mathbb{H}_R,
\end{align*}
where $k \in \{1,2\}$ and $j \in \mathbb{Z}^+$, the series converges on $\c_R$. Then
\begin{align*}
f(t) = \sum_{n=j}^{\infty}\frac{a_{n+k}}{\Gamma\left(\frac{n}{k}+1\right)}t^{n/k},
\end{align*}
and the series converges for all $t\in \c$.
\end{corollary}

\section{Key results for hyperexponential processes}\label{sect:key_hyp}
\subsection{Overview}
Recall that the Laplace exponent $\psi(z)$ of a hyperexponential process $X$ is a rational function of the form \eqref{eq:main} with real poles  $\{\rho_{\ell}\}_{1\leq \ell \leq N}$ and $\{-\hat{\rho}_{\ell}\}_{1\leq \ell \leq \hat{N}}$, which we assume are arranged according to increasing magnitude. Further, solutions of the equation $\psi(z) = q$ are of particular interest. It is not difficult to show that these are always real when $q > 0$; we denote the positive (resp. negative) solutions by $\{\zeta_{\ell}\}_{1\leq \ell \leq M}$ ($\{-\hat{\zeta}_{\ell}\}_{1\leq \ell \leq \hat{M}}$) where $M=N$ or $M=N+1$ (resp. $\hat{M}= \hat{N}$ or $\hat{M} = \hat{N} +1$) and $M$ (resp. $\hat{M}$) is determined by the values $\sigma$ and $a$. If we want to emphasize the argument $q$ we will write, for example, $\zeta_{\ell}(q)$.\\ \\
\noindent Importantly, we have the \emph{interlacing property}
\begin{align}\label{eq:inter}
 -\hat{\rho}_{\hat{N}} <-\hat{\zeta}_{\hat{N}}<\ldots-\hat{\rho}_1 <-\hat{\zeta}_2 < -\hat{\rho}_1 < -\hat{\zeta}_1 < 0 <  \zeta_1 < \rho_1 < \zeta_2 < \rho_2 < \ldots < \zeta_{N} < \rho_{N}.
\end{align}
When $\sigma \neq 0$ we have $M =N +1$ and $\hat{M} = \hat{N} + 1$, that is we have two additional solutions $-\hat{\zeta}_{\hat{M}}$ and $\zeta_M$ occurring to the left and right of $-\hr_{\hN}$ and $\rho_N$ respectively. Otherwise, if $\sigma = 0$ and $a > 0$ we have $M =N +1$ and $\hM = \hN$, and if $\sigma = 0$ and $a < 0$ then $M = N$ and $\hM = \hN + 1$. Again, in the cases where $-\hz_{\hM}$ and $\zeta_M$ represent additional solutions, they will occur to the left and right of $-\hr_{\hN}$ and $\rho_N$ respectively. Finally if both $\sigma=0$ and $a=0$ then $\hM = \hN$ and $M = N$ and there are no additional solutions. These ideas are illustrated in Figure \ref{fig:psiz} where the case $\sigma=0$, $a > 0$ is shown.\\ \\ 
\noindent \noindent We are interested in the distribution of the random variable $X_{\eq}$, which represents the process $X$ at the random time $\eq$. Here $\eq$ is an exponential random variable independent of $X$ with mean $q^{-1}$. To determine the distribution of $X_{\eq}$ we take the Laplace transform, which has the form
\begin{align*}
F(z;q) := \e[e^{zX_{\eq}}] = \frac{q}{q-\psi(z)},\quad -\hat{\zeta}_1 < \re(z) < \zeta_1,
\end{align*} 
and observe that, like $\psi(z)$, $F(z;q)$ extends to a rational function on $\c$. From the discussion above, it is clear that $F(z;q)$ has simple zeros at points $\{\rho_{\ell}\}_{1\leq \ell \leq N}$ and $\{-\hr_{\ell}\}_{1\leq \ell \leq \hN}$ and simple poles at points $\{\zeta_{\ell}\}_{1\leq \ell \leq M}$ and $\{-\hz_{\ell}\}_{1\leq \ell \leq \hM}$.
To simplify the presentation of what follows, we adopt the notation
\begin{align}\label{eq:gamma}
\gamma := \frac{1}{\left(q + \sum_{n=1}^{N}a_n + \sum_{n=1}^{\hat{N}}\hat{a}_n\right)},
\end{align}
and adhere to the conventions $\rho_0 := \hat{\rho}_0 := 0$ and $\rho_{N+1} := \hat{\rho}_{\hat{N}+1} := +\infty$.
\begin{figure}[!h]
\centering
\def\svgscale{0.35}
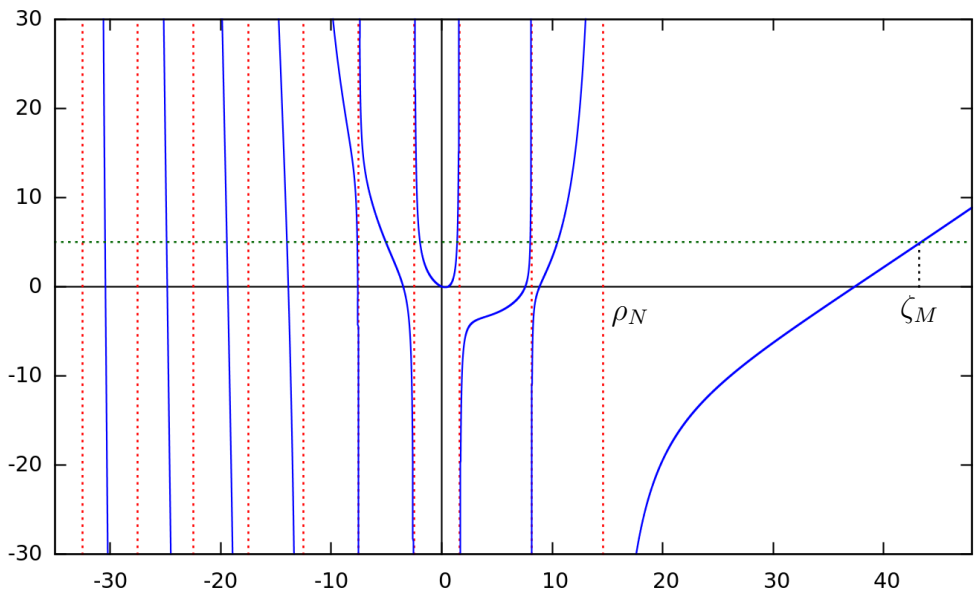
\caption{A plot of $\psi(z)$. The vertical dashed lines show the position of the poles while the horizonal dashed line shows the position of $q$.}\label{fig:psiz}
\end{figure}
\begin{theorem}\label{theo:atq}
The random variable $X_{\eq}$ has distribution
\begin{align*}
\p\left(X_{\eq} \in \d x\right) = q\alpha\delta_0(\d x) + q\left(\ind(x > 0)\sum_{\ell=1}^M\frac{e^{-\zeta_{\ell} x}}{\psi'(\zeta_{\ell})} - \ind(x < 0)\sum_{\ell=1}^{\hat{M}}\frac{e^{\hat{\zeta}_{\ell} x}}{\psi'(-\hat{\zeta}_{\ell})}\right)\d x,
\end{align*}
where $\alpha = \gamma$ when $\sigma^2$ and $a$ are both zero, and $\alpha = 0$ otherwise.
\end{theorem}
\begin{proof} A partial fraction decomposition of $F(z;q)$ has the form
\begin{align*}
F(z;q) = c + \sum_{\ell=1}^{M}\frac{\textnormal{Res}(F,\zeta_{\ell})}{z - \zeta_{\ell}} + \sum_{\ell=1}^{\hM}\frac{\textnormal{Res}(F,\hz_{\ell})}{z+\hz_{\ell}} = c - \sum_{\ell=1}^{M}\frac{q}{\psi'(\zeta_{\ell})(z - \zeta_{\ell})} - \sum_{\ell=1}^{\hM}\frac{q}{\psi'(-\hz_{\ell})(z+\hz_{\ell})}
\end{align*}
for some constant $c$. To determine $c$ we can take the limit $z \rightarrow +\infty$ on the left and right of the previous expression. It is easy to see that $c$ is non-zero only when both $\sigma$ and $a$ are identically zero and takes the value $q\gamma$ in this case. Now, inverting the Laplace transform gives the result.
\end{proof}
\begin{remark}
The reader may wish to compare Theorem \ref{theo:atq} with Theorem 2 (v) in \cite{KuzKyPa2011} which gives the analogous result for the meromorphic family of processes.\rdemo
\end{remark}
\noindent Let us conclude this section by defining two specific parameter sets which we will use for numeric examples throughout the remainder of the paper. These are taken from \cite{Jeannin} and are derived from stock market data by means of approximation of a VG and NIG process. We define:\\ \\
Parameter Set 1:
\begin{align*}
\{a_1,\,a_2,\,a_3,\,a_4,\,a_5,\,a_6,\,a_7\} &= \{0.0255,\,0.0255,\,0.0510,\,0.3060,\,0.6120,\,0.9690,\,3.1110\} \\
\{\rho_1,\,\rho_2,\,\rho_3,\,\rho_4,\,\rho_5,\,\rho_6,\,\rho_7\} &= \{5,\,10,\,15,\,25,\,30,\,60,\,80\}\\
\{\hat{a}_1,\,\hat{a}_2,\,\hat{a}_3,\,\hat{a}_4,\,\hat{a}_5,\,\hat{a}_6,\,\hat{a}_7\} &= \{0.3200,\,0.1920,\,0.7040,\,0.5120,\,0.6400,\,2.5600,\,1.4720\} \\
\{\hat{\rho}_1,\,\hat{\rho}_2,\,\hat{\rho}_3,\,\hat{\rho}_4,\,\hat{\rho}_5,\,\hat{\rho}_6,\,\hat{\rho}_7\} &= \{5,\,10,\,15,\,25,\,30,\,60,\,80\},
\end{align*}
Parameter Set 2:
\begin{align*}
\{a_1,\,a_2,\,a_3,\,a_4,\,a_5,\,a_6,\,a_7\} &= \{0.0066,\,0.0154,\,0.4620,\,0.1760,\,0.5720,\,0.4180,\,0.5500\} \\
\{\rho_1,\,\rho_2,\,\rho_3,\,\rho_4,\,\rho_5,\,\rho_6,\,\rho_7\} &= \{5,\,10,\,15,\,25,\,30,\,60,\,80\}\\
\{\hat{a}_1,\,\hat{a}_2,\,\hat{a}_3,\,\hat{a}_4,\,\hat{a}_5,\,\hat{a}_6,\,\hat{a}_7\} &= \{0.0300,\,0.2700,\,0.9300,\,0.9300,\,0.3000,\,0.2400,\,0.3000\} \\
\{\hat{\rho}_1,\,\hat{\rho}_2,\,\hat{\rho}_3,\,\hat{\rho}_4,\,\hat{\rho}_5,\,\hat{\rho}_6,\,\hat{\rho}_7\} &= \{2,\,5,\,10,\,30,\,50,\,80,\,100\}.
\end{align*}
\subsection{Solutions of $\psi(z) = q$}
\noindent In this section, using the tools developed in Section \ref{sect:working}, we derive series expansions of  solutions of $\psi(z)=q$, for $q \in \c$ such that $\vert q \vert$ is large. We remark that the numbering, order, and multiplicity of the solutions $\{\zeta_{\ell}\}_{1\leq \ell \leq M}$ and $\{-\hat{\zeta}_{\ell}\}_{1\leq \ell \leq \hat{M}}$, which we defined for real $q > 0$, in particular the interlacing property \eqref{eq:inter}, lose their meaning if we allow $q$ to be a complex number. We show here, however, that for each $\zeta_{\ell}$, $1 \leq \ell \leq M$ (resp. $\hat{\zeta}_\ell$, $1 \leq \ell \leq \hat{M}$) and $\vert q \vert$ large enough, there exists a solution $z_{\ell}$ (resp. $\hat{z}_{\ell}$) of $\psi(z) = q$, which is an analytic function on  $\c_R^+$ and corresponds to exactly one element of $\{\zeta_{\ell}\}_{\ell=1}^M \cup \{-\hat{\zeta}_{\ell}\}_{\ell=1}^{\hat{M}}$, namely $\zeta_{\ell}$ (resp. $\hat{\zeta}_{\ell}$). That is, each $\zeta_{\ell}$ (resp. $\hat{\zeta}_{\ell}$) may be extended to an analytic function, which also solves $\psi(z)=q$,  provided $q \in \c^+$ is large enough. We will denote this extended solution using the same notation, i.e. $\zeta_{\ell}$ (resp. $\hat{\zeta}_{\ell}$).\\ \\
\noindent  To begin, let us define the functions 
\begin{align*}
g(z;{\ell}) := \psi(z+\rho_{\ell}),\quad 1 \leq {\ell} \leq N,\quad \hat{g}(z;{\ell}):=\psi(z-\hat{\rho}_{\ell}),\quad 1 \leq {\ell} \leq \hat{N},\quad \text{ and } \quad h(z) := \psi\left(\frac{1}{z}\right).
\end{align*}
Our first goal is to derive the Laurent series expansions of these functions at $0$. Starting with $h(z)$, and using the summation formula for geometric series, we obtain
\begin{align}\label{eq:psi_coef}
h(z) &= \frac{\sigma^2}{2}z^{-2} + az^{-1} -\left( \sum_{\ell=1}^{N}\frac{a_{\ell}}{1-\rho_{\ell}z} + \sum_{\ell=1}^{\hat{N}}\frac{\hat{a}_{\ell}}{1 + \hat{\rho}_{\ell}z}\right)\nonumber\\
&= \frac{\sigma^2}{2}z^{-2} + az^{-1} + \sum_{n=0}^{\infty}\eta_nz^n,
\end{align}
\begin{align}\label{eq:eta}
\eta_n :=-\left( \sum_{\ell=1}^{N}a_{\ell}(\rho_{\ell})^n +   (-1)^{n}\sum_{\ell=1}^{\hat{N}}\hat{a}_{\ell}(\hat{\rho}_{\ell})^n\right),
\end{align}
which is valid on the deleted neighbourhood $0 < \vert z \vert < \min\left\{\frac{1}{\rho_N},\frac{1}{\hat{\rho}_{\hat{N}}}\right\}$. Going forward, we will write $\{h_n\}_{n\geq -2}$ to refer to the coefficients of the expansion of $h(z)$. Continuing, again via the technique of geometric series, we have
\begin{align}\label{eq:g_series}
g(z;{\ell}) &= \frac{\sigma^2}{2}(z+\rho_{\ell})^2 + a(z + \rho_{\ell}) + (z + \rho_{\ell})\left(\sum_{n=1}^{N}\frac{a_n}{\rho_n - (z + \rho_{\ell})} - \sum_{n=1}^{\hat{N}}\frac{\hat{a}_n}{\hat{\rho}_n + (z + \rho_{\ell})}\right) \nonumber\\
&= -a_{\ell}\rho_{\ell}z^{-1} + \left(\frac{\sigma^2\rho_{\ell}^2}{2} + a\rho_{\ell} + \omega_0\right) + \left(\sigma^2 \rho_{\ell} + a + \omega_1\right)z\nonumber \\ &\hphantom{= -a_{\ell}\rho_{\ell}z^{-1} }+ \left(\frac{\sigma^2}{2} + \omega_2\right)z^2 + \sum_{n=3}^{\infty}\omega_n z^{n}, 
\end{align}
\begin{gather}\label{eq:omega}
\omega_n := \omega_{{\ell},n} := \begin{dcases}
\sum_{\substack{i=1 \\ i \neq {\ell}}}^{N}\frac{a_i\rho_{\ell}}{\rho_i - \rho_{\ell}} - \sum_{i=1}^{\hat{N}}\frac{\hat{a}_i\rho_{\ell}}{\hat\rho_i + \rho_{\ell}} - a_{\ell} & n=0\\
 \sum_{\substack{i=1 \\ i \neq {\ell}}}^{N}\frac{a_i\rho_i}{(\rho_i - \rho_{\ell})^{n+1}} +  (-1)^{n}\sum_{i=1}^{\hat{N}}\frac{\hat{a}_i\hat{\rho}_i}{(\hat{\rho}_i + \rho_{\ell})^{n+1}}& n \in \mathbb{N}
\end{dcases},
\end{gather}
which is valid for $ 0 < \vert z \vert < \min\left\{ \vert \rho_{{\ell}-1} - \rho_{\ell} \vert,\,\vert \rho_{{\ell}+1} - \rho_{\ell} \vert\right\}$.
\noindent In an analogous fashion we obtain
\begin{align}\label{eq:gh_series}
\hat{g}(z;\ell) &= \frac{\sigma^2}{2}(z - \hat{\rho}_{\ell})^2 + a(z - \hat{\rho}_{\ell}) + (z - \hat{\rho}_{\ell})\left(\sum_{n=1}^{N}\frac{a_n}{\rho_n - (z - \hat{\rho}_{\ell})} - \sum_{n=1}^{\hat{N}}\frac{\hat{a}_n}{\hat{\rho}_n + (z - \hat{\rho}_{\ell})}\right) \nonumber\\
&= \hat{a}_{\ell}\hat{\rho}_{\ell}z^{-1} + \left(\frac{\sigma^2\hat{\rho}_{\ell}^2}{2} - a\hat{\rho}_{\ell} + \hat{\omega}_0\right) + \left(a - \sigma^2 \hat{\rho}_{\ell}  + \hat{\omega}_1\right)z \nonumber\\ &
\hphantom{=\hat{a}_{\ell}\hat{\rho}_{\ell}z^{-1}} + \left(\frac{\sigma^2}{2} + \hat{\omega}_2\right)z^2 + \sum_{i=3}^{\infty}\hat{\omega}_n z^{n}, 
\end{align}
\begin{gather}\label{eg:omega_hat}
\hat{\omega}_n := \hat{\omega}_{\ell,n} := \begin{dcases}
\sum_{\substack{i=1 \\ i \neq {\ell}}}^{\hat{N}}\frac{\hat{a}_i\hat{\rho}_{\ell}}{\hat{\rho}_i - \hat{\rho}_{\ell}} - \sum_{i=1}^{N}\frac{a_i\hat{\rho}_{\ell}}{\rho_i + \hat{\rho}_{\ell}} - \hat{a}_{\ell} & n=0\\
 \sum_{i=1}^{N}\frac{a_i\rho_i}{(\rho_i + \hat{\rho}_{\ell})^{n+1}}+  (-1)^{n} \sum_{\substack{i=1 \\ i \neq {\ell}}}^{\hat{N}}\frac{\hat{a}_i\hat{\rho}_i}{(\hat{\rho}_i - \hat{\rho}_{\ell})^{n+1}} & n \in \mathbb{N},
\end{dcases}
\end{gather}
which is valid for $0 < \vert z \vert <  \min\{\vert \hat{\rho}_{{\ell}-1} - \hat{\rho}_{\ell}\vert,\vert\hat{\rho}_{{\ell}+1} - \hat{\rho}_{\ell} \vert\}$. Going forward we will write $\{g_{\ell,n}\}_{n\geq-1}$ or just $\{g_n\}_{n\geq -1}$ for the coefficients of the expansion of $g(z;\ell)$ and we will adopt analogous notation for $\hat{g}(z;\ell)$. Further, we remind the reader that $_{k}\bar{f}_n$ denotes the first $n+k+1$ elements of the collection $\{f_n\}_{n\geq k}$, e.g. the first $n+k+1$ coefficients of a series, and that $s_n,\,m_n,\,r_n,\,p_n$ and $l_n$ denote the coefficients of series derived via summation, multiplication, reciprocation, composition with $c^z$, and inversion respectively. See respectively \eqref{def:-kfn}, \eqref{def:sn}, \eqref{def:mn}, \eqref{def:rn}, \eqref{def:pn}, \eqref{def:ln} in Section \ref{sect:working}. 
\begin{theorem}\label{theo:extra_roots}\ \\
\begin{enumerate}[(i)]
\item Assume $\sigma > 0$. Then there exists $R > 0$ such that for $q \in \c^+_{R}$
\begin{align}\label{eq:zeta_M}
\zeta_M = \sum_{n=-1}^{\infty}\frac{b_n}{q^{n/2}}, \quad\text{ and }\quad \hat{\zeta}_{\hat{M}} = -\sum_{n=-1}^{\infty}(-1)^{n}\frac{b_n}{q^{n/2}},
\end{align}
where
\begin{gather*}
b_n = r_n(_1\bar{c}_{n+2}),\quad c_i = l_i(_2\bar{d}_{i+1}),\quad\text{ and }\quad d_j = r_j(_{-2}\bar{h}_{j-4}),
\end{gather*}
and the series in \eqref{eq:zeta_M} converge on $\c_{R}$.
\item If $\sigma = 0$ and $a > 0$ (respectively $a < 0$) then there exists $R > 0$ such that for $q \in \c_{
R}$,
\begin{align*}
\zeta_M = \sum_{n=-1}^{\infty}\frac{b_n}{q^{n}}\quad\left(\text{resp. } \hat{\zeta}_{\hat{M}} = -\sum_{n=-1}^{\infty}\frac{b_n}{q^n}\right),
\end{align*}
where
\begin{gather*}
b_n = r_n(_1\bar{c}_{n+2}),\quad c_i = l_i(_1\bar{d}_{i}),\quad\text{ and }\quad d_j = r_j(_{-1}\bar{h}_{j-2}),
\end{gather*}
and the series converges on $\c_{R}$.
\end{enumerate}
\end{theorem}

\begin{proof}
We prove only the case $\sigma> 0$, the other cases can be proven in essentially the same way. We want to solve $\psi(z) = q$ for large $\vert q \vert$. Setting $z = \frac{1}{v}$ and $q = \frac{1}{w}$ we see this is equivalent to solving $w = \frac{1}{h(v)}$ for small $\vert w \vert$. We first determine the expansion of the reciprocal of $h(v)$. This yields, 
\begin{align*}
w = \sum_{j=2}^{\infty}r_j(_{-2}\bar{h}_{j-4})v^{j} = \sum_{j=2}^{\infty}d_jv^{j},
\end{align*}
which is valid near zero. Now we invert the series; choosing the principal branch of the square root function we get as one solution
\begin{align*}
v = \sum_{i=1}^{\infty}l_i(_2\bar{d}_{i+1})w^{i/2} = \sum_{i=1}^{\infty}c_iw^{i/2},
\end{align*}
which is again valid near zero. Reciprocating the resulting series once more gives
\begin{align}\label{eq:z_zeta}
z =  \sum_{n=-1}^{\infty}r_n(_1\bar{c}_{n+2})w^{n/2} = \sum_{n=-1}^{\infty}b_nw^{n/2} = \sum_{n=-1}^{\infty}b_n\left(\frac{1}{q}\right)^{n/2}.
\end{align}
Since the series in $w$ converges on $B_0(1/R^{1/2})$ for some $R >0 $, the right-hand side of \eqref{eq:z_zeta} converges on $\c_{R}$; further, we have
\begin{align*}
\sum_{n=-1}^{\infty}b_n\left(\frac{1}{q}\right)^{n/2} = \sum_{n=-1}^{\infty}\frac{b_n}{q^{n/2}}
\end{align*}
on $\c^+_{R}$ and the right-hand side of this expression clearly also converges on $\c_{R}$.
Thus we have $z = (2q/\sigma^2)^{1/2} + b_0 + O(q^{-1/2})$ so that $z\rightarrow +\infty$ as $q \rightarrow +\infty$. We know from the interlacing property \eqref{eq:inter} that all other solutions of $\psi(z) = q$ take values strictly less than $\rho_{N}$ for real $q$; therefore $z$ must correspond to $\zeta_{M}$. Similar reasoning and choosing the other branch of the square root yields the result for $\hat{\zeta}_{\hat{M}}$.
\end{proof}
\noindent The following corollaries follow more or less directly from Theorem \ref{theo:extra_roots} and the the series manipulations discussed in Section \ref{sect:working}; the proofs are left to the reader. 
\begin{corollary}\label{cor:with_sig}\ \\
Assume $\sigma > 0$ and that $\{b_n\}_{n\geq-1}$ are as in Theorem \ref{theo:extra_roots} (i).
\begin{enumerate}[(i)]
\item There exits $R > 0$ such that
\begin{align*}
\frac{1}{\zeta_M} = \sum_{i=1}^{\infty}\frac{c_i}{q^{i/2}},\quad\text{ and }\quad \frac{1}{\hat{\zeta}_{\hat{M}}} = -\sum_{i=1}^{\infty}(-1)^i\frac{c_i}{q^{i/2}},\quad q \in \c^+_{R},
\end{align*}
where the $\{c_i\}_{i\geq 1}$ are as in Theorem \ref{theo:extra_roots} (i), and the series converge on $\c_{R}$.
\item There exits $R > 0$ such that for $c_0 := b_0 - 1$, $c_n := b_n$, $n\neq 0$ and $d_n := (-1)^{n}c_n$, $n\geq-1$, we have
\begin{align*}
\frac{1}{\zeta_M - 1} = \sum_{n=1}^{\infty}\frac{r_n(_{-1}\bar{c}_{n-2})}{q^{n/2}},\quad\text{ and }\quad \frac{1}{\hat{\zeta_M} + 1} =-\sum_{n=1}^{\infty}\frac{r_n(_{-1}\bar{d}_{n-2})}{q^{n/2}},\quad q \in \c^+_{R},
\end{align*}
and the series converge on $\c_{R}$.
\item There exists $R > 0$ such that for $c_n := (-1)^{n}b_n$, $n\geq-1$, we have
\begin{align*}
k^{-\zeta_M} = k^{-(2q/\sigma^2)^{1/2}}\sum_{n=0}^{\infty}\frac{p_n(k,-\bar{b}_n)}{q^{n/2}},\quad \text{ and }\quad k^{\hat{\zeta}_{\hat{M}}} = k^{(2q/\sigma^2)^{1/2}}\sum_{n=0}^{\infty}\frac{p_n(k,-\bar{c}_n)}{q^{n/2}},\quad q \in \c^+_{R},
\end{align*}
and the series converge on $\c_{R}$.
\item The exists $R$ such that we have 
\begin{align*}
\zeta_M' = \sum_{n=1}^{\infty}\left(\frac{2-n}{2}\right)\frac{b_{n-2}}{q^{n/2}},\quad\text{ and }\quad -\hat{\zeta}_{\hat{M}}' = \sum_{n=1}^{\infty}(-1)^n\left(\frac{2-n}{2}\right)\frac{b_{n-2}}{q^{n/2}},\quad q \in \c^+_{R},
\end{align*}
and the series converge on $\c_{R}$
\end{enumerate}
\end{corollary}
\begin{corollary}\label{cor:without_sig}\ \\
Assume $\sigma = 0$, $a > 0$ (resp. $a < 0$), and that $\{b_n\}_{n\geq-1}$ are as in Theorem \ref{theo:extra_roots} (ii).
\begin{enumerate}[(i)]
\item There exists $R > 0$ such that
\begin{align*}
\frac{1}{\zeta_M} = \sum_{i=1}^{\infty}\frac{c_i}{q^i}\quad \left(\text{resp. }\frac{1}{\hat{\zeta}_{\hat{M}}} = -\sum_{i=1}^{\infty}\frac{c_i}{q^i}\right),\quad q \in \c_{R},
\end{align*}
where the $\{c_i\}_{i\geq 1}$ are as in Theorem \ref{theo:extra_roots} (ii) and the series converges on $\c_{R}$.
\item  There exists $R > 0$ such that
\begin{align*}
\frac{1}{\zeta_M - 1} = \sum_{n=1}^{\infty}\frac{r_n(_{-1}\bar{c}_{n-2})}{q^n}\quad\left(\text{resp. } \frac{1}{\hat{\zeta}_{\hat{M}} + 1} = -\sum_{n=1}^{\infty}\frac{r_n(_{-1}\bar{c}_{n-2})}{q^n}\right),\quad q \in \c_{R},
\end{align*}
where $c_0 := b_0-1$, $c_n := b_n$ for $n\neq 0$, and the series converges on $\c_{R}.$
\item There exists $R > 0$ such that
\begin{align*}
k^{-\zeta_M} = k^{-q/a}\sum_{n=0}^{\infty}\frac{p_n(k,-\bar{b}_n)}{q^n}\quad \left(\text{ resp. } k^{\hat{\zeta}_{\hat{M}}} = k^{-q/a}\sum_{n=0}^{\infty}\frac{p_n(k,-\bar{b}_n)}{q^n}\right),\quad q \in \c_{R},
\end{align*}
where the series converges on $\c_{R}$.
\item There exists $R > 0$ such that
\begin{align*}
\zeta_M' = \sum_{n=0}^{\infty}\left(1-n\right)\frac{b_{n-1}}{q^n}\quad\left(\text{resp. }-\hat{\zeta}_{\hat{M}}' = \sum_{n=0}^{\infty}\left(1-n\right)\frac{b_{n-1}}{q^n}\right),\quad q \in \c_{R},
\end{align*}
where the series converges on $\c_{R}$.
\end{enumerate}
\end{corollary}
\noindent Now, we turn our attention to the remaining solutions $\{\zeta_{\ell}\}_{1 \leq \ell \leq N}$, and $\{\hat{\zeta}_{\ell}\}_{1 \leq \ell \leq \hat{N}}$. It will be convenient to write the coefficients of the Laurent series representations of $g(z;\ell)$ and $\hat{g}(z;\ell)$ as simply $\{g_n\}_{n\geq-1}$ and $\{\hat{g}_n\}_{n\geq -1}$ without reference to the index $\ell$. However, the reader should note that when describing a series expansion of $\zeta_{\ell}$ in terms of $\{g_n\}_{n\geq-1}$ we always intend the coefficients of $g(z;\ell)$ and similarly for $\hat{\zeta}_{\ell}$.
\begin{theorem}\label{theo:normal_roots}
For each $1 \leq \ell \leq N$ (resp. $1 \leq \ell \leq \hat{N}$), there exists $R > 0$ such that
\begin{align*}
\zeta_{\ell} = \sum_{n=0}^{\infty}\frac{b_n}{q^{n}} \quad \left(\text{ resp. }\hat{\zeta}_{\ell} = -\sum_{n=0}^{\infty}\frac{\hat{b}_n}{q^{n}}\right),\quad q \in \c_{R},
\end{align*}
where $b_0 = \rho_{\ell}$ (resp. $\hat{b}_0 = -\hat{\rho}_{\ell}$),\quad 
\begin{align*}
b_n = l_{n}(_1\bar{c}_n),\quad \text{ and } \quad c_i = r_i(_{-1}\bar{g}_{i-2})\quad \left(\text{ resp. }\hat{b}_n = l_{n}\left(_1\bar{\hat{c}}_n\right),\quad\text{ and }\quad \hat{c}_i = r_i\left(_{-1}\bar{\hat{g}}_{i-2}\right)\right),
\end{align*}
and the series converges on $\c_{R}$.
\end{theorem}
\begin{proof}
We prove only the result for the positive solutions, the proof of the result for the negative solutions is analogous. Our goal is to solve $g(z;\ell) = q$ for large $\vert q \vert$ or, by making the change of variables $q = \frac{1}{w}$, to solve $\frac{1}{g(z;\ell)} = w$ for small $\vert w \vert$. We first determine the expansion of the reciprocal of $g(z;\ell)$. This yields, for $z$ near zero,
\begin{align*}
w = \sum_{i=1}^{\infty}r_n(_{-1}\bar{g}_{i-2})z^{i} = \sum_{i=1}^{\infty}c_iz^{i}.
\end{align*}
Now we invert the series giving
\begin{align}\label{eq:basic_z}
z = \sum_{n=1}^{\infty}l_n(_1\bar{c}_n)w^{n} = \sum_{n=1}^{\infty}\frac{b_n}{q^{n}},
\end{align}
where the left-hand side converges on $B_0(1/R)$ for some $R > 0$, so that the right-hand side converges on $\c_{R}$.\\ \\
\noindent Now we let $v = z + \rho_{\ell}$. It remains to show that $\rho_{\ell - 1} < v < \rho_{\ell}$, or equivalently that $z < 0$, for $q>0$ large enough. However, if this were not the case, i.e. if for every $M > 0$ we could find a $q > M$ such that $z \geq 0$, then, since $z \rightarrow 0$, for at least one such $z$ we would have either $g(z;\ell) = \psi(v) < 0 \neq q$ or $g(z;\ell) = \psi(\rho_{\ell}) \neq q$; this can be verified from the definition of $\psi(z)$ in \eqref{eq:main} and also visually from Figure \ref{fig:psiz}. By the interlacing property \eqref{eq:inter} we conclude that $v$ must correspond to $\zeta_{\ell}$ .
\end{proof}
\noindent The following corollary is straightforward to prove using Theorem \ref{theo:normal_roots} and the rules for series manipulations; we leave the proof to the reader. Corollary \ref{cor:norm} is the analog of Corollaries \ref{cor:with_sig} and \ref{cor:without_sig}.
\begin{corollary}\label{cor:norm} For each $1 \leq \ell \leq N$ (resp. $1 \leq \ell \leq \hat{N}$), and for each of (i), (ii), (iii), and (iv), there exists $R > 0$ such that the series representations hold and converge for all $q \in \c_{R}$. Throughout $\{b_n\}_{n\geq0}$ (resp. $\{\hat{b}_n\}_{n\geq 0}$) are as defined in Theorem \ref{theo:normal_roots}.
\begin{enumerate}[(i)]
\item 
\begin{align*}
\frac{1}{\zeta_{\ell}} = \sum_{n=0}^{\infty}\frac{r_n(\bar{b}_n)}{q^{n}}\quad \left(\text{resp. }\frac{1}{\hat{\zeta}_{\ell}} = -\sum_{n=0}^{\infty}\frac{r_n\left(\bar{\hat{b}}_n\right)}{q^{n}}\right)
\end{align*}
\item Under the assumption $\rho_1 > 1$ we have
\begin{align*}
\frac{1}{\zeta_{\ell} - 1} = \sum_{n=0}^{\infty}\frac{r_n(\bar{c}_n)}{q^{n}}\quad \left(\text{resp. }\frac{1}{\hat{\zeta}_{\ell} + 1} = -\sum_{n=0}^{\infty}\frac{r_n\left(\bar{\hat{c}}_n\right)}{q^{n}}\right),
\end{align*}
where $c_0 = b_0 - 1$ and $c_n = b_n$ (resp. $\hat{c}_0 = \hat{c}_0 - 1$ and $\hat{c}_n = \hat{b}_n$) otherwise.
\item For any $k > 0$
\begin{align*}
k^{-\zeta_{\ell}} = \sum_{n=0}^{\infty}\frac{p_n(k,-\bar{b}_n)}{q^{n}}\quad \left(\text{resp. }k^{\hat{\zeta}_{\ell}} = \sum_{n=0}^{\infty}\frac{p_n(k,-\bar{\hat{b}}_n)}{q^{n}}\right).
\end{align*}
\item
\begin{align*}
\zeta'_{\ell} = \sum_{n=2}^{\infty}(1-n)\frac{b_{n-1}}{q^n}\quad\left(\text{resp. } -\hat{\zeta}'_{\ell} =\sum_{n=2}^{\infty}(1-n)\frac{\hat{b}_{n-1}}{q^n}\right)
\end{align*}
\end{enumerate}
\end{corollary}
\begin{example}[\textbf{Exotic option pricing via the Laplace transform}]\label{ex:barr}
We let $X$ be the process defined by Parameter Set 1 and setting $\sigma = 0.042$ and $a = 0.141875$. Further, we define $\overline{X}$ to be the running supremum process, i.e. $\overline{X}_t := \sup_{s\leq t}X_s$ and consider the problem of computing
\begin{align*}
D(t):= D(t;k) := e^{-rt}\p\left(\overline{X}_t < \log(k)\right)
\end{align*}
for some fixed $t,\,k,\,r > 0$. Readers may recognize that $D(t)$ represents the value of a European up and out digital option. Using the Wiener-Hopf factorization it is straightforward to show that
\begin{align*}
\mathcal{L}\{e^{rt}D(t)\}(q) = \frac{1 - \sum_{\ell=1}^{8}\beta_{\ell}k^{-\zeta_{\ell}}}{q},\quad \text{ where } \quad \beta_{\ell} := \beta_{\ell}(q) := \frac{ \prod_{\ell=k}^{7}\left(1 - \frac{\zeta_{\ell}(q)}{\rho_{k}}\right)}{\prod_{\substack{1 \leq k \leq 8 \\ k \neq \ell }}\left(1 - \frac{\zeta_{\ell}(q)}{\zeta_k(q)} \right)},\,1 \leq \ell \leq 8;
\end{align*}
see \cite{Jeannin} for further details of the derivation. To recover $e^{rt}D(t)$ we need to compute
\begin{align*}
\frac{1}{2\pi i}\int_{c + i\r}e^{tq}\mathcal{L}\{e^{rt}D(t)\}(q)\d q = \Re\left(\frac{e^{ct}}{\pi}\int_{0}^{\infty}e^{itu}\frac{1 - \sum_{\ell=1}^{8}\beta_{\ell}(c + iu)k^{-\zeta_{\ell}(c+iu)}}{c + iu}\d u\right)
\end{align*}
where $c > 0$ and as usual $i := \sqrt{-1}$. The integral on the right necessitates evaluation via a numerical quadrature rule; in particular, for each step in the quadrature we will need to compute $\{\zeta_{\ell}(c + iu)\}_{1 \leq \ell \leq 8}$ numerically. While there are several good numerical approaches to doing this, this ``root finding" is typically the most time intensive part of the algorithm. To speed up the computation, we can replace the numerical procedure for finding $\{\zeta_{\ell}(c + iu)\}_{1 \leq \ell \leq 8}$ by the (truncated) series expansions of Theorems \ref{theo:extra_roots} and \ref{theo:normal_roots} once $u$ is big enough (typically choosing large $c$ does not yield a reliable algorithm).\\ \\
\noindent To test this idea, we use the root finding algorithm proposed in Section 5 of \cite{Kuz2010a}, which can be described as follows: a) differentiate the identity $\psi(\zeta_{\ell}) = q$ with respect to $u$ and b) solve the resulting ODE for $\zeta_{\ell}$ at each point in the discretization of $u$ using a numerical method like the midpoint method. At each step, sharpen the estimate by applying two or three iterations of Newton's root finding algorithm. To evaluate the integral we use Filon's quadrature \cite{Filon,Fosdick}. For an overview of the combined approach i.e. root finding together with Filon's quadrature see also Chapter 5.2 in \cite{Hackmann}.\\ \\
Fixing $t = 0.25$, $k = 1.1$, $r=0.03$ and $c = 0.5$ we compute $D(0.25)$ once by the method described in the previous paragraph (Method 1) and once by replacing the numerical root finding technique by the truncated series for $ u  > 80$ (Method 2). For this we use the series up to order $q^{-10}$ for all $1 \leq \ell \leq 8$. Depending our need for accuracy, we can vary the number of discretization steps and upper limit of integration when we apply Filon's quadrature. The two methods are compared in Table \ref{tab:opt_calc}. We see that Method 1 takes between three to four-and-a-half times as long as Method 2 to find the roots and roughly two to three times as long in total. We find the prices are identical up to 12 decimal places. Not included in the time for Method 2 is the time needed to compute the coefficients of the series expansions of $\{\zeta_{\ell}\}_{1\leq \ell \leq 8}$. In total, however, this is negligible: approximately 0.1 seconds for all of the coefficients. \\ \\
\noindent To get a further sense for the accuracy of the series expansions, we also present some graphical experiments in Figures \ref{fig:ex_root_graph} and \ref{fig:r6_graph}. In all cases we truncate the series at order $q^{-10}$. In Figure \ref{fig:ex_root_graph_plot} we see plots of $\zeta_8(q)$ (blue $\circ$) and $\hat{\zeta}_8(q)$ (red $\square$), where $q = 0.5 + iu$, $u \in [80,150]$. The corresponding errors $\vert \psi(\zeta_8) - q \vert$ and $\vert \psi(-\hat{\zeta}_8) - q \vert$ are shown in Figure \ref{fig:ex_root_graph_error}. Further we compute $\zeta'_8$ and $\hat{\zeta}'_8$ in Figure \ref{fig:ex_root_graph_der} and show the errors $\vert\zeta'_8 - 1/\psi'(\zeta_8)\vert$ and $\vert-\hat{\zeta}'_8 - 1/\psi'(-\hat{\zeta}_8)\vert$ in Figure \ref{fig:ex_root_graph_der_error}. Note that once $q$ is large enough, so that $\zeta_{\ell}$ is an analytic function, we may differentiate both sides of $\psi(\zeta_{\ell}) = q$ with respect to $q$ showing that $\zeta'_{\ell} = 1/\psi(\zeta_{\ell})$ and similarly that $-\hat{\zeta}'_{\ell} = 1/\psi(-\hat{\zeta}_{\ell})$; this is the basis for the error calculations for the derivatives. The same information for $\zeta_6$ and $\hat{\zeta}_6$ is shown in Figure \ref{fig:r6_graph}.\demo
\renewcommand{\arraystretch}{1.5}
\begin{table}
\begin{center}
\begin{tabular}{ | p{1.8cm} | p{1.8cm} | p{1.8cm} | p{1.8cm} | p{1.8cm} | p{1.8cm} | p{1.8cm} | p{1.8cm} |}
  \hline 
 Disc. Steps & Integ. Limit & $D(0.25)$ Method 1 & $D(0.25)$ Method 2 & Time roots Method 1  & Time roots Method 2  & Total time Method 1 & Total time Method 2 \\ \hline  \hline
$10^5$ & $10^3$ & 0.896525 & 0.896525 & 2.852 & 0.956  & 3.364 & 1.460\\ \hline
$10^6$  & $10^4$ & 0.896764 & 0.896764 & 27.428 & 6.332 & 32.488 & 11.392\\ \hline
$10^7$ & $10^5$ & 0.896865 & 0.896865 & 273.224 & 56.300 & 323.876 & 106.976 \\ \hline
\end{tabular}
\caption{$D(0.25)$ by two methods with varying number of discretization steps and upper limits of integration. Time is in seconds. ``Time roots" refers to the time spent by the algorithm with computing the roots; ``Total time" is the total time needed to compute $D(0.25)$. }\label{tab:opt_calc}
\end{center}
\end{table}
\begin{figure}[!h]
\centering
\subfloat[]
{\label{fig:ex_root_graph_plot}\includegraphics[scale=0.60]{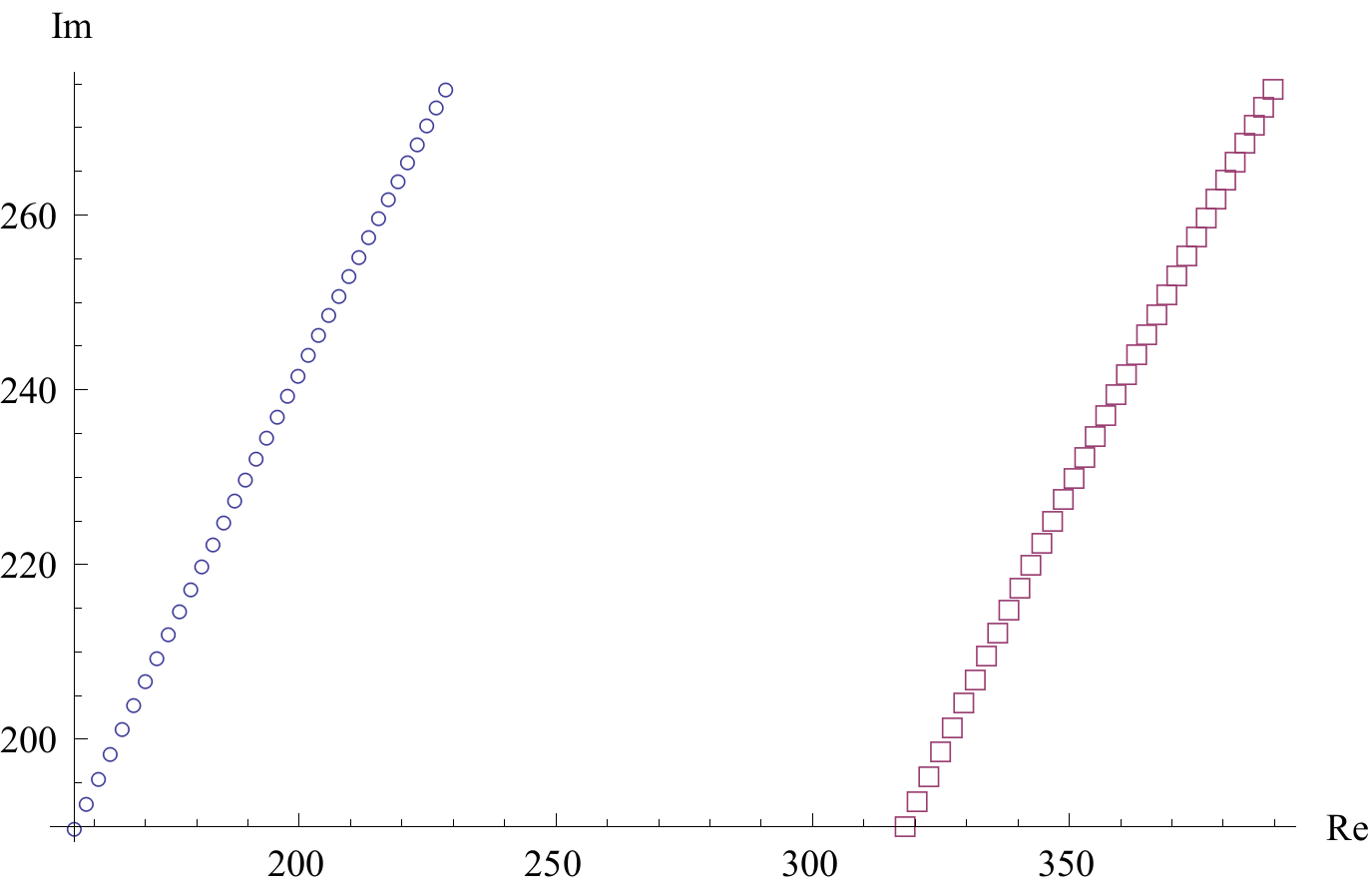}}
\subfloat[]
{\label{fig:ex_root_graph_error}\includegraphics[scale=0.60]{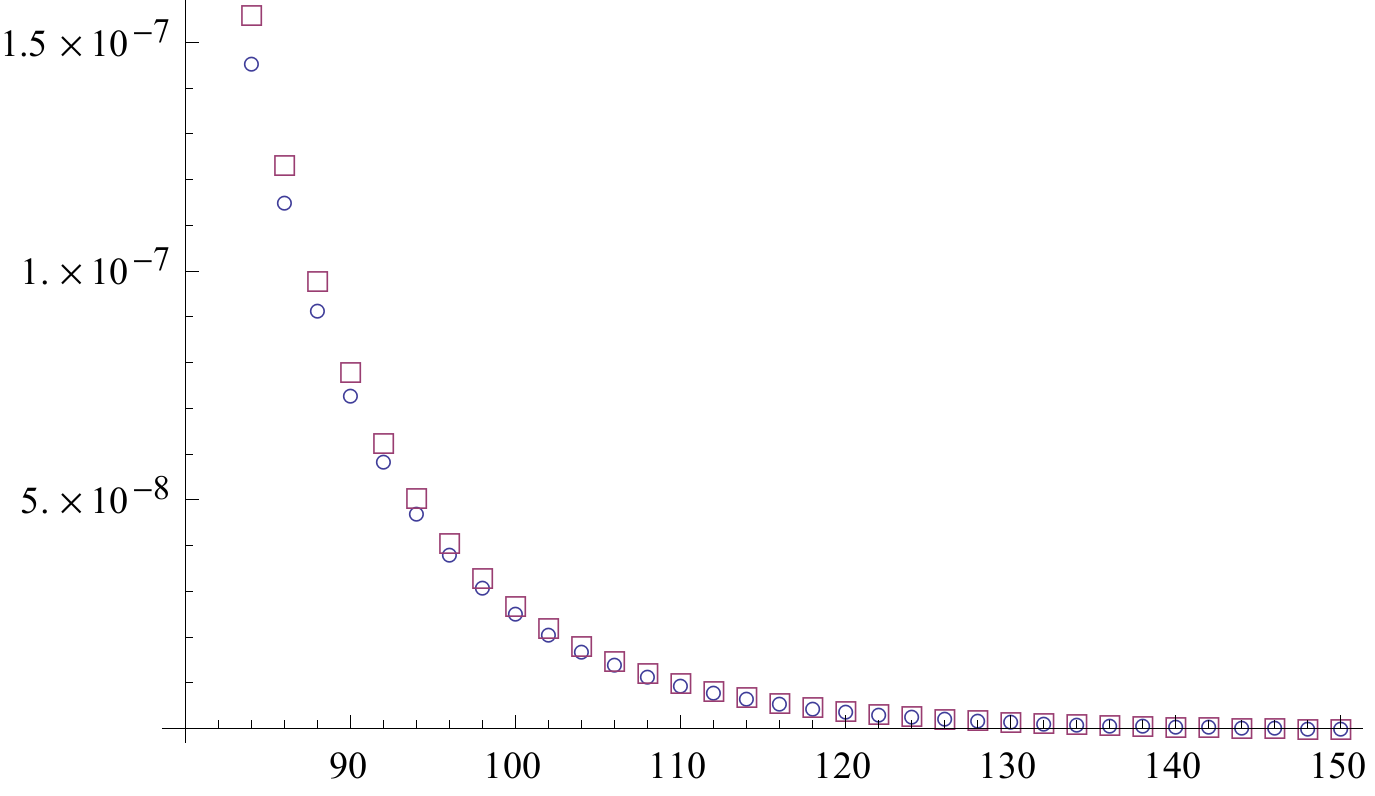}}\\
\subfloat[]
{\label{fig:ex_root_graph_der}\includegraphics[scale=0.60]{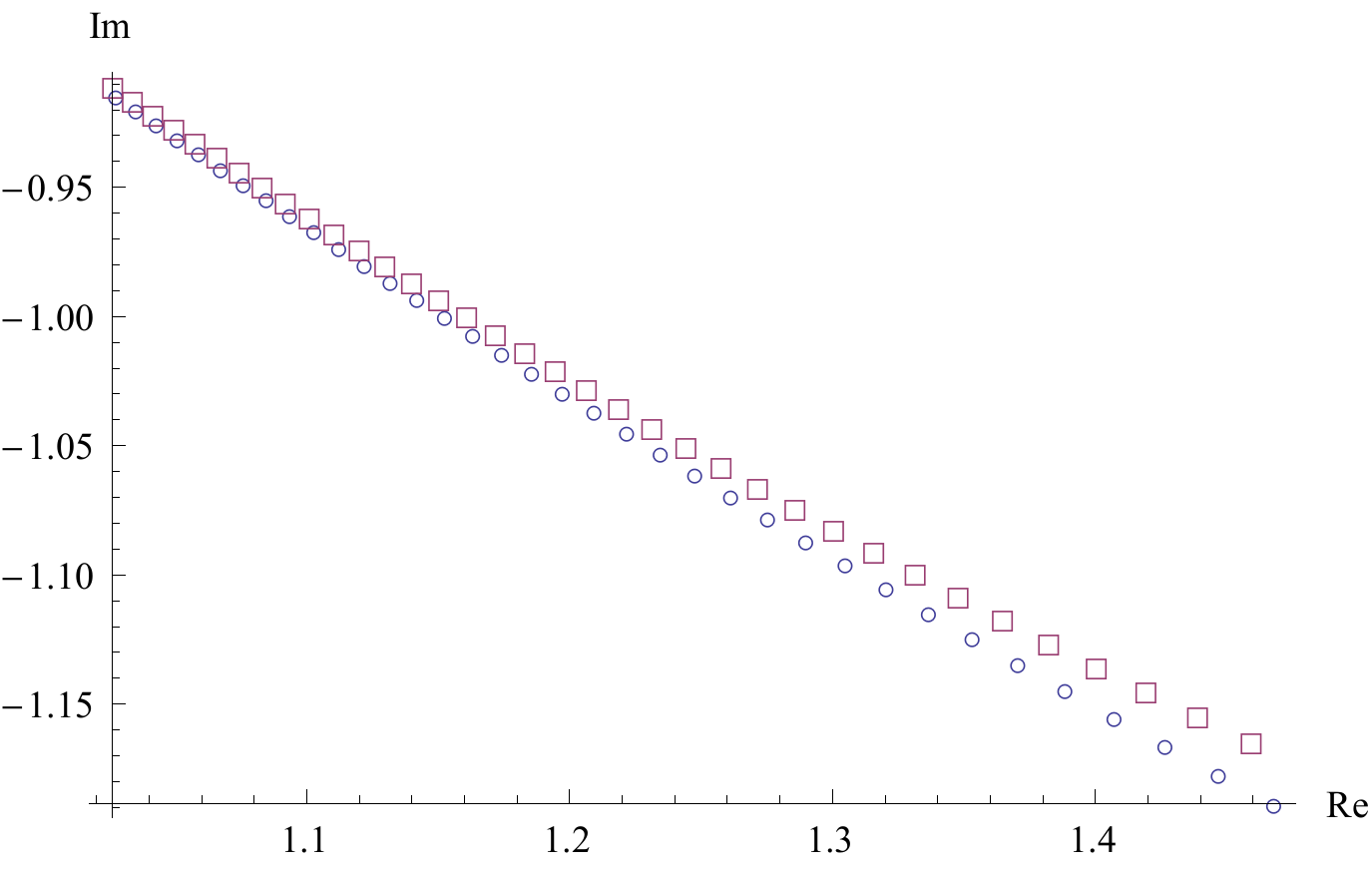}}
\subfloat[]
{\label{fig:ex_root_graph_der_error}\includegraphics[scale=0.60]{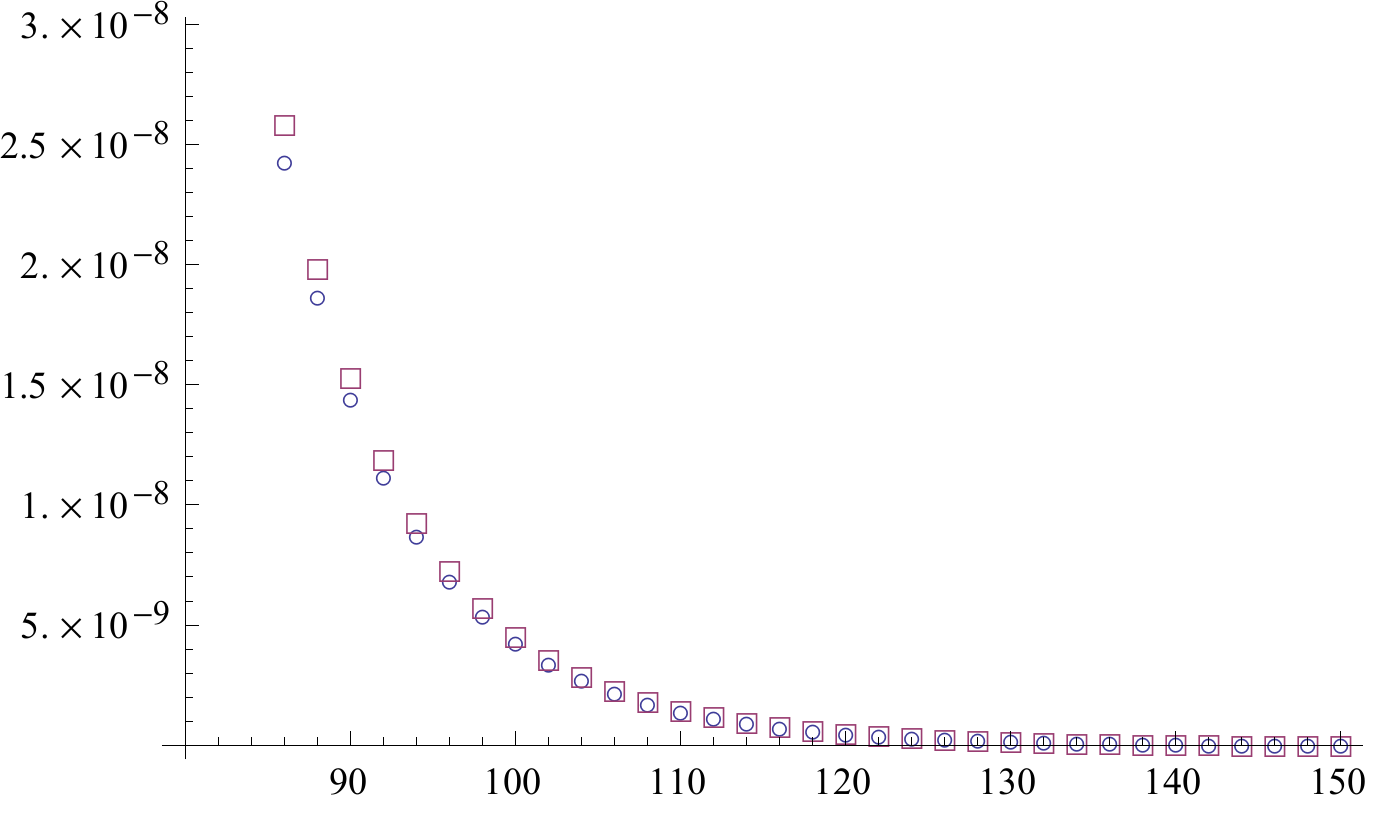}}
\caption{Plots of $\zeta_8$ (blue $\circ$) and $\hat{\zeta_8}$ (red $\square$) are shown in \ref{fig:ex_root_graph_plot}, errors in \ref{fig:ex_root_graph_error}. Plots of $\zeta'_8$ and $\hat{\zeta'_8}$ are shown in \ref{fig:ex_root_graph_der}, errors in \ref{fig:ex_root_graph_der_error}.}
\label{fig:ex_root_graph}
\end{figure}
\begin{figure}[!h]
\centering
\subfloat[]
{\label{fig:r6_graph_plot}\includegraphics[scale=0.60]{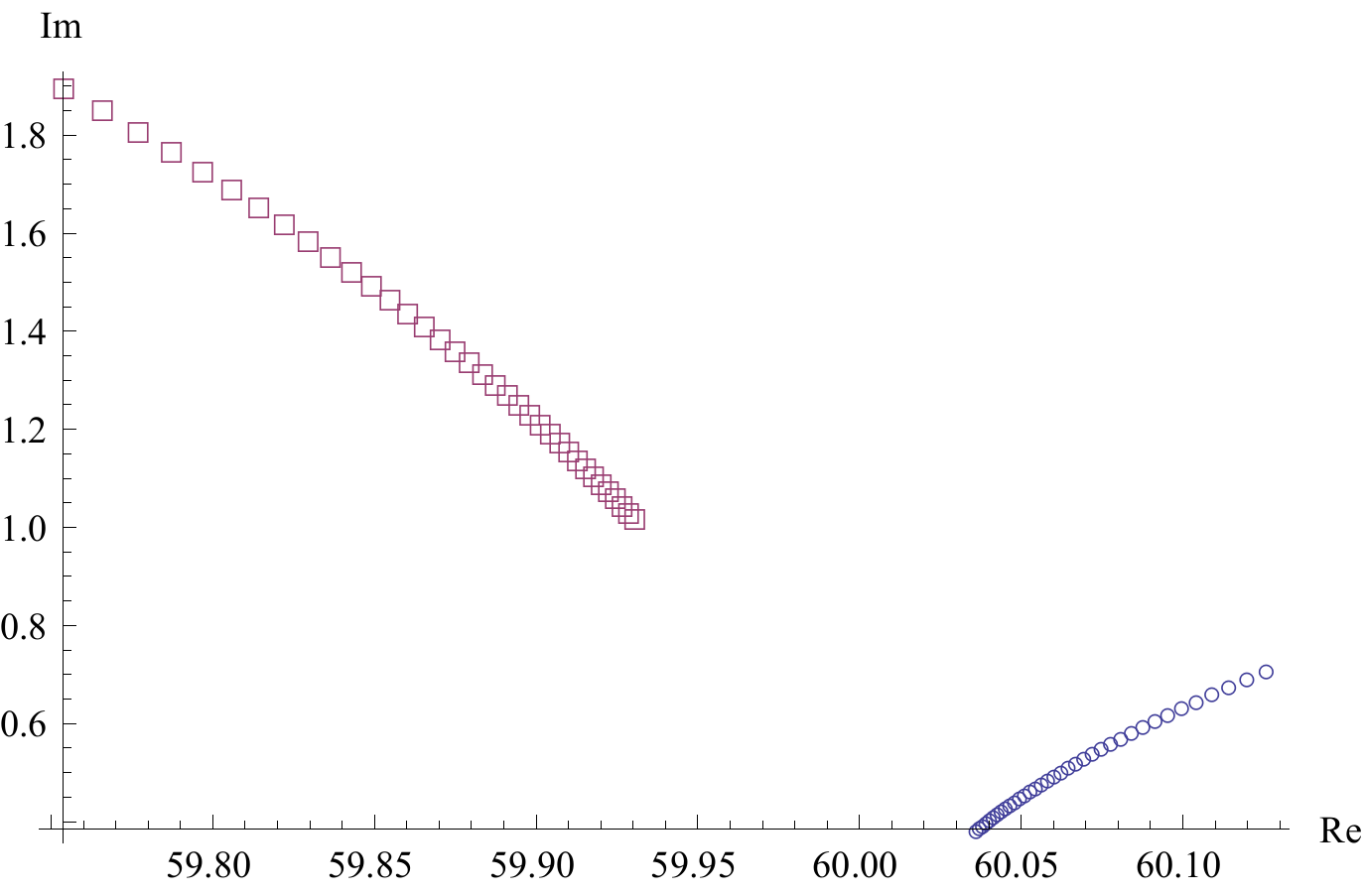}}
\subfloat[]
{\label{fig:r6_graph_error}\includegraphics[scale=0.60]{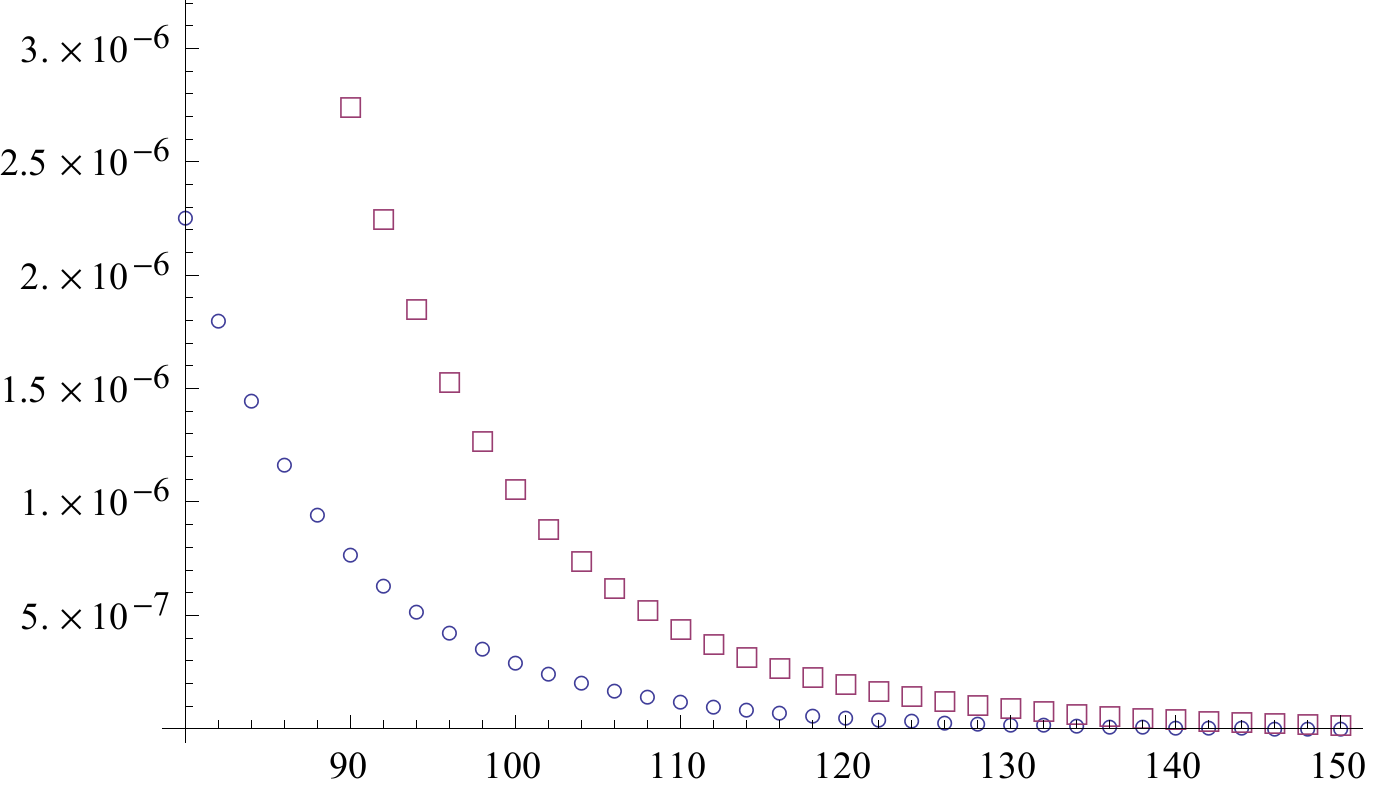}}\\
\subfloat[]
{\label{fig:r6_graph_der}\includegraphics[scale=0.60]{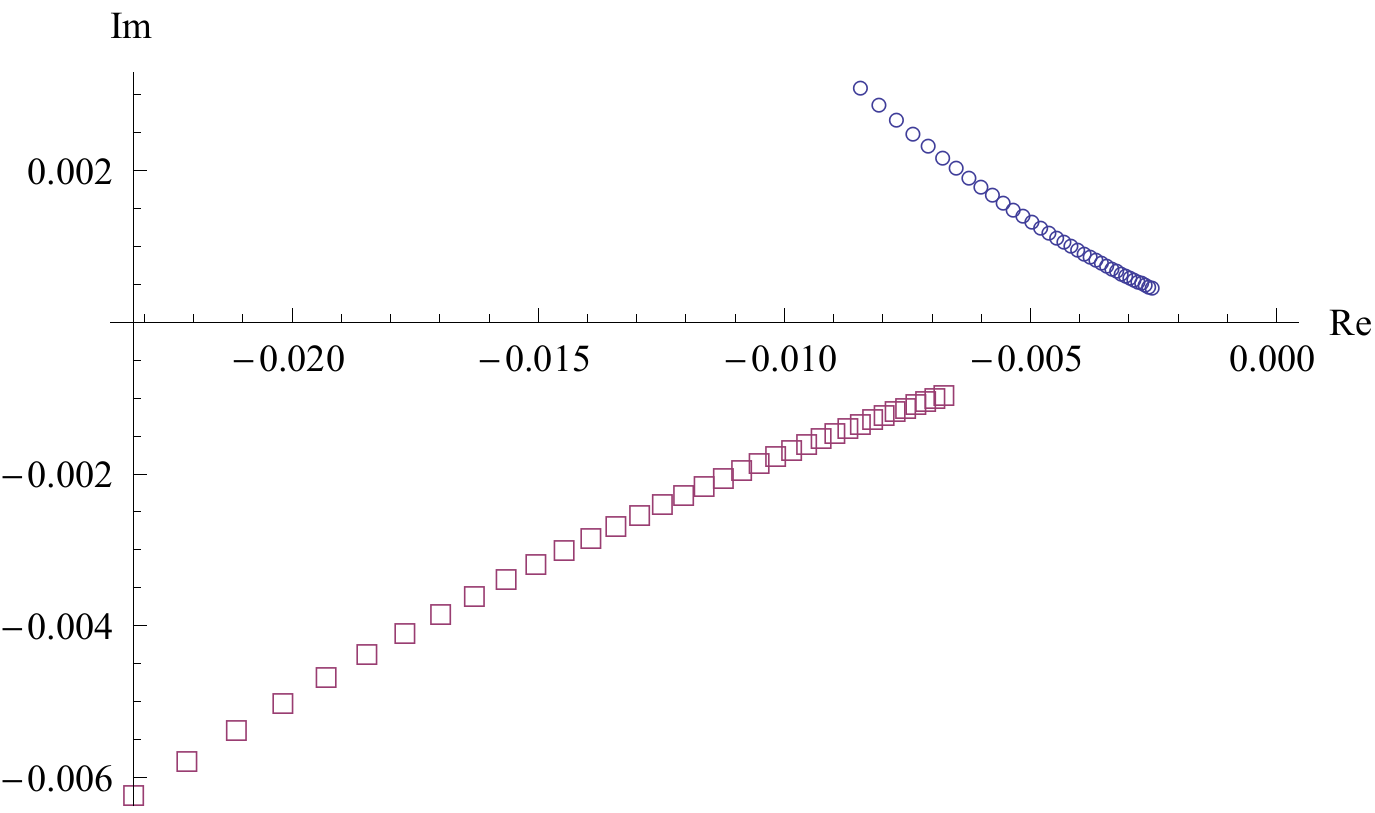}}
\subfloat[]
{\label{fig:r6_graph_der_error}\includegraphics[scale=0.60]{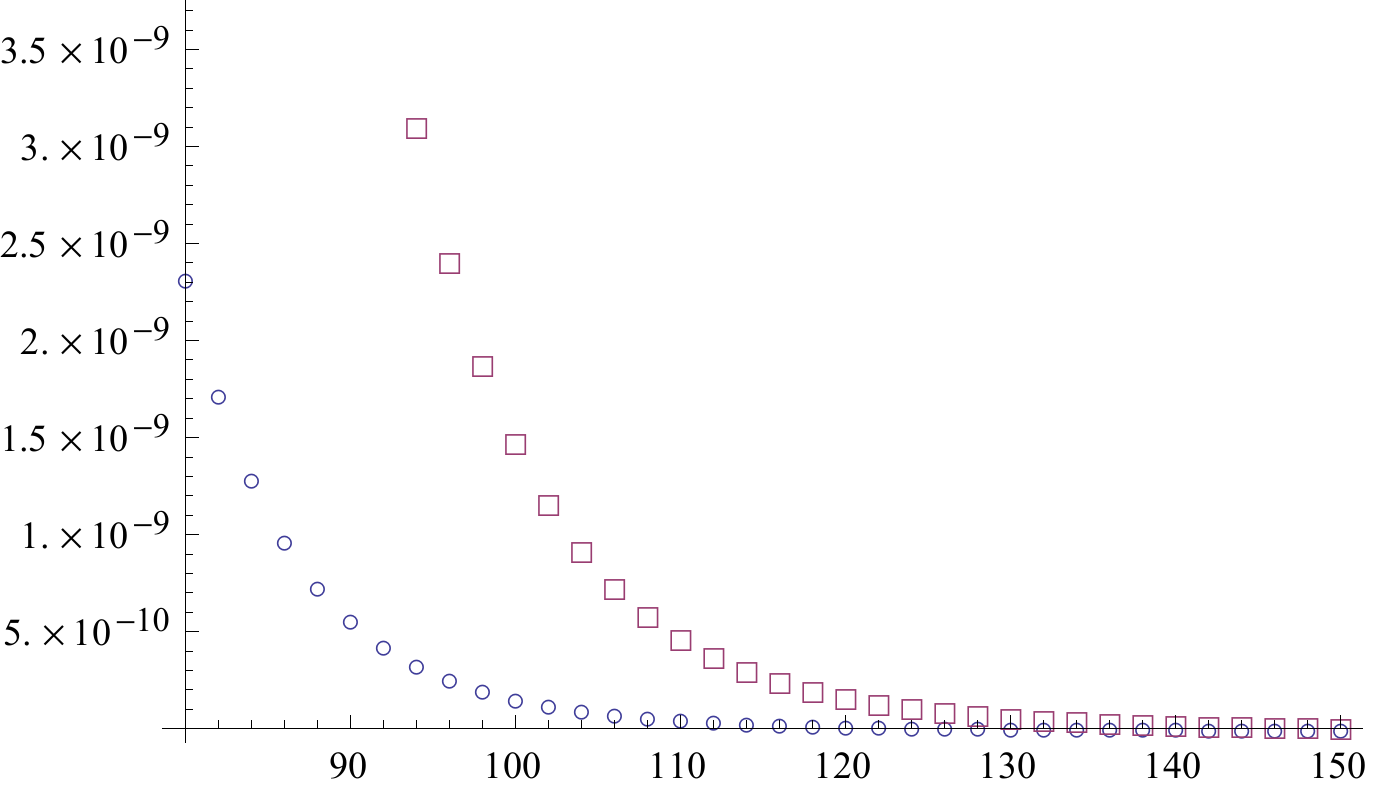}}
\caption{Plots of $\zeta_6$ (blue $\circ$) and $\hat{\zeta_6}$ (red $\square$) are shown in \ref{fig:r6_graph_plot}, errors in \ref{fig:r6_graph_error}. Plots of $\zeta'_6$ and $\hat{\zeta'_6}$ are shown in \ref{fig:r6_graph_der}, errors in \ref{fig:r6_graph_der_error}.}
\label{fig:r6_graph}
\end{figure}
\end{example}
\section{Analytic formulas for European option prices and Greeks}\label{sect:option}
In this section we model a stock price $S$ by an exponentiated hyperexponential process. That is, we set
\begin{align*}
S_t := S_0e^{X_t},\quad t > 0,
\end{align*}
where $S_0 > 0$ is the stock price at time zero, $X$ is a hyperexponential process satisfying $\psi(1) = r > 0$ and $\rho_1 > 1$, and $r$ represents some fixed interest rate. This latter condition ensures that the discounted process $e^{-r t}S_t$ is a martingale, in other words, that the so-called risk neutral condition is satisfied. We will simply say that $X$ \emph{fulfils the risk-neutral condition} whenever $\psi(1) = r > 0$ and $\rho_1 > 1$. \\ \\
\noindent It is well-known that the theoretical price of a European call (resp. put) option in this scenario is then given by the expression
\begin{align*}
C := C(T) := e^{-r T}\e[(S_T - K)^+],\quad\left(\text{resp. }P := P(T) := e^{-r T}\e[(K -S_T)^+\right])
\end{align*}
where $K > 0$ is the strike price, $T > 0$ is the date of expiry, and $(x)^+ = \max\{0,x\}$, $x \in \r$. Writing
\begin{align}\label{eq:def_of_f}
C(T) = e^{-r T}\times S_0 \times f(T), \quad f(t) := \e[(e^{X_t} - k)^+],\quad k := \frac{K}{S_0}
\end{align}
allows us to focus on the difficult part of evaluating $C(T)$, namely evaluating $f(t)$, without carrying around additional terms. Let $w(t)$ and $W(q)$ be defined analogously for the put $P(T)$. If we take the Laplace transform of $f(t)$ we see that
\begin{align}\label{eq:laplace}
F(q) := \int_{0}^{\infty}e^{-qt}f(t)\d t = \frac{1}{q}\e[(e^{X_{\eq}} - k)^+] 
\end{align}
where $\eq$ is again an exponential random variable with mean $q^{-1}$. Since we know the distribution of $X_{\eq}$ from Theorem \ref{theo:atq} we can work out an explicit expression for $F(q)$ in terms of $\{\zeta_n\}_{1\leq n \leq M}$ and $\{\hat{\zeta}_n\}_{1\leq n \leq \hat{M}}$ as we do below. Assuming first that $k > 1$, i.e. that the call option is out of the money (OTM), we have
\begin{align*}
F(q) &= \frac{1}{q}\int_{\r}(e^{x} - k)^{+}\p(X_{\eq} \in \d x)
= \int_{\log(k)}^{\infty}(e^x - k)\sum_{{\ell}=1}^{M}\frac{e^{-\zeta_{\ell}x}}{\psi'(\zeta_{\ell})}\d x = k\sum_{{\ell}=1}^{M}\frac{k^{-\zeta_{\ell}}}{\psi'(\zeta_{\ell})\zeta_{\ell}(\zeta_{\ell} - 1)}.
\end{align*}
Now if $\vert q \vert$ is large enough, we may replace $1/\psi'(\zeta_{\ell})$ by $\zeta'_{\ell}$ and similarly $1/\psi'(-\hat{\zeta}_{\ell})$ by $-\hat{\zeta}_{\ell}$. Therefore, for $q > 0$ large enough, when $k > 1$ we have
\begin{align}\label{eq:medium}
F(q) = k\sum_{{\ell}=1}^{M}\frac{\zeta'_{\ell}k^{-\zeta_{\ell}}}{\zeta_{\ell}(\zeta_{\ell} - 1)},
\end{align}
when $k = 1$ we have
\begin{align}\label{eq:easy}
F(q) = \sum_{{\ell}=1}^{M}\frac{\zeta'_{\ell}}{\zeta_{\ell}(\zeta_{\ell} - 1)},\end{align}
and finally, when $k < 1$, we have
\begin{align}\label{eq:hard}
W(q) = k\sum_{{\ell}=1}^{\hat{M}}\frac{\hat{\zeta}'_{\ell}k^{\hat{\zeta}_{\ell}}}{\hat{\zeta}_{\ell}(\hat{\zeta}_{\ell} + 1)}.
\end{align}
\noindent We remark that although Formulas \ref{eq:medium}--\ref{eq:hard} are derived under the assumption that $q$ is real -- see \eqref{eq:laplace} -- the properties of Laplace transforms ensure that since \eqref{eq:laplace} holds for all $q > 0$ the Laplace transform of the functions $f(t)$ and $w(t)$, i.e. $\mathcal{L}\{f\}(q)$ and $\mathcal{L}\{w\}(q)$, are well defined and analytic on $\mathbb{H}$ (see Theorems 3.4 and 6.1 in \cite{doetsch}). Further, Corollaries \ref{cor:with_sig}, \ref{cor:without_sig}, and \ref{cor:norm} (iv) as well as the fact that none of $\{\zeta_{\ell} - 1,\, \zeta_{\ell}\}_{1 \leq \ell \leq M}$, or $\{\hat{\zeta}_{\ell} + 1,\,\hat{\zeta}_{\ell}\}_{1 \leq \ell \leq \hat{M}}$ are zero for $q$ with $\re(q)$ large enough, show that Formulas \ref{eq:medium}--\ref{eq:hard} actually define analytic functions on a half-plane $\mathbb{H}_{q_0}$ for some $q_0 > 0$. But then, $\mathcal{L}\{f\}(q)$ and $F(q)$ (resp. $\mathcal{L}\{w\}(q)$ and $W(q)$) are both analytic on $\mathbb{H}_{q_0}$  and they agree on $\mathbb{H}_{q_0} \cap \r $. Invoking an argument of analytic continuation (see Corollary 3.2.4.1 in \cite{krantz} ) shows that in fact $\mathcal{L}\{f\}(q) = F(q)$ (resp. $\mathcal{L}\{w\}(q) = W(q)$) on $\mathbb{H}_{q_0}$, i.e. Formulas \ref{eq:medium}--\ref{eq:easy} (resp. \ref{eq:hard}) are valid in the half-plane $\mathbb{H}_{q_0}$ and are equal to the Laplace transform of the function $f(t)$ (resp. $w(t)$) there. \\ \\
\noindent We now use this fact, and the series expansions derived in the previous section, to develop convergent series expansions of the option price $C$, starting with the simplest case $k = 1$, i.e. the at the money (ATM) case. To facilitate notation we define
\begin{align}\label{eq:xi_varsig}
\xi_{\ell} := 1/\zeta_{\ell},&\quad\text{ and }\quad \varsigma_{\ell} := 1/(\zeta_{\ell} -1),\quad\text{ for }\quad 1\leq\ell\leq M,\quad\text{ and }\nonumber\\
\hat{\xi}_{\ell} := 1/\hat{\zeta}_{\ell},&\quad\text{ and }\quad\hat{\varsigma}_{\ell} := 1/(\hat{\zeta}_{\ell} + 1),\quad\text{ for }\quad 1\leq \ell \leq \hat{M},
\end{align}
and remind the reader that the $n$-th coefficient of the sum (resp. product) of a finite number of series is denoted $s_n$ (resp. $m_n$).
\begin{lemma}\label{lem:at_the_money}
Suppose that $X$ fulfils the risk neutral condition and $S_0 = K$. 
\begin{enumerate}[(i)]
\item If $\sigma > 0$ then there exists $R > 0$ such that the following equality holds for $q \in \c^+_{R}$  and  the series converges for $q \in \c_{R}$:
\begin{align*}
F(q) = \sum_{n=3}^{\infty}\frac{b_n}{q^{n/2}},
\end{align*}
where 
\begin{align*}
b_n = \begin{dcases}
s_{n/2} + m_{M,n} & n \text{ even}\\
m_{M,n} & n \text{ odd}
\end{dcases},
\end{align*}
$s_n = s_n(m_{1,n},\,m_{2,n},\,\ldots,\,m_{N,n})$, and $m_{\ell,n} := m_{n}(\bar{\xi}_{\ell,n},\,\bar{\varsigma}_{\ell,n},\,\bar{\zeta'}_{\ell,n})$ for $1 \leq \ell \leq M$.
\item If $\sigma = 0$ then there exists $R > 0$ such that the following equality holds and  the series converges for $q \in \c_{R}$:
\begin{align*}
F(q) = \sum_{n=2}^{\infty}\frac{b_n}{q^{n}}
\end{align*}
where $b_n = s_n(m_{1,n},\,m_{2,n},\,\ldots,\,m_{M,n})$, and $\{m_{\ell,n}\}_{1\leq \ell \leq M}$ are as in (i).
\end{enumerate}
\end{lemma}
\begin{proof}
Use Formula \ref{eq:easy} and apply Corollaries \ref{cor:with_sig}, \ref{cor:without_sig}, and \ref{cor:norm}.
\end{proof}
\begin{theorem}[\textbf{ATM call option price}]\label{theo:at_the_money}
Suppose that $X$ fulfils the risk neutral condition and $S_0 = K$.
\begin{enumerate}[(i)]
\item Let $\sigma > 0$ and $\{b_n\}_{n\geq 3}$ be as in Lemma \ref{lem:at_the_money} (i). Then
\begin{align}\label{eq:at_the_money_gauss}
C(T) = e^{-r T}S_0\sum_{n=1}^{\infty}b_{n+2}\frac{T^{n/2}}{\Gamma\left(\frac{n}{2} + 1\right)},\quad T \geq 0,
\end{align}
and the series on the right-hand side converges for $T \in \c$.
\item Let $\sigma = 0$ and $\{b_n\}_{n\geq 2}$ be as in Lemma \ref{lem:at_the_money} (ii). Then
\begin{align}\label{eq:at_the_money_no_gauss}
C(T) = e^{-r T}S_0\sum_{n=1}^{\infty}b_{n+1}\frac{T^{n}}{n!},\quad T \geq 0,
\end{align}
and the series on the right-hand side converges for $T \in \c$.
\end{enumerate}
\end{theorem}
\begin{proof}
Apply Lemma \ref{lem:at_the_money} and Corollary \ref{cor:main} together with the fact that $f(t)$ is a continuous function; this latter claim follows from the stochastic continuity of L\'{e}vy processes.
\end{proof}
\begin{example}[\textbf{ATM implied volatility}]\label{ex:bs}
We consider a simple case: $B$ is a hyperexponential process without jumps, i.e. $N = \hat{N} = 0$, such that $\sigma > 0$, i.e. $B$ yields the classic Black-Scholes model. Further let us assume $K = S_0 = 1$ and $r = 0$. Using Theorem \ref{theo:at_the_money}, we may compute the option price, which we denote $C_B(\sigma,T)$, symbolically up to a reasonably high order. For example, up to order $T^{9/2}$ we have:
\begin{align}\label{eq:imp_vol_series}
C_B(\sigma,T) =& \frac{\sigma }{\sqrt{2 \pi }}T^{1/2}-\frac{\sigma ^3 }{24 \sqrt{2 \pi }}T^{3/2}+\frac{\sigma ^5 }{640\sqrt{2 \pi } }T^{5/2}-\frac{\sigma ^7 }{ 21504\sqrt{2 \pi }}T^{7/2} + O(T^{9/2}).
\end{align}
We find that we can compute the first forty terms of this series symbolically in about 0.2 seconds. Given the form of the first few terms, we might ask if \eqref{eq:imp_vol_series} may also be interpreted as a series in powers of $\sigma T^{1/2}$; the fact that it can, is obvious from the Black-Scholes formula.\\ \\
\noindent Now let us suppose that we have a hyperexponential process $X$ with jumps, i.e. at least one of $N$ or $\hat{N}$ is not zero, with Gaussian component $\sigma > 0$, and that $X$ satisfies the risk neutral condition. Let $C_X(T)$ denote the option price under process $X$, again under the assumptions that $r = 0$ and $K=S_0=1$. Further, let us implicitly define a function $\hat{\sigma}(T)$ as that value, which yields
\begin{align}\label{eq:imp_comp}
C_B(\hat{\sigma}(T),T)= C_X(T), \quad T > 0.
\end{align} 
We are interested in finding an asymptotic expansion of $\hat{\sigma}(T)$, which known as the at-the-money implied volatility -- for a proper discussion of this concept consult Section 4 of \cite{tankcar} or Section 11 of \cite{Cont}.
\noindent To do this, we expand the left-hand side of \eqref{eq:imp_comp} in $s:=\hat{\sigma}(T)T^{1/2}$ (c.f. \eqref{eq:imp_vol_series}) and invert the series to solve for $s$ in powers of $w:=C_X(T)$, i.e.
\begin{align}\label{eq:inv_series_impv}
s = \hat{\sigma}(T)T^{1/2} =& \sqrt{2 \pi } w + \frac{\pi ^{3/2}}{6 \sqrt{2}} w^3 + \frac{7 \pi ^{5/2}}{240\sqrt{2} } w^5 \nonumber \\ &+ \frac{127 \pi ^{7/2} }{20160 \sqrt{2}}w^7+\frac{4369 \pi ^{9/2}}{2903040 \sqrt{2}} w^9+\frac{34807 \pi ^{11/2} }{91238400 \sqrt{2}}w^{11} + O(w^{13}).
\end{align}
Then we expand $C_X(T)$ using Theorem \ref{theo:at_the_money}, 
\begin{align}\label{eq:exp_series}
w = C_X(T)&= \frac{\sigma }{\sqrt{2 \pi }} \sqrt{T}  +\left(\frac{ 2 a+\sigma ^2}{4} + \sum_{\ell=1}^{N}\frac{a_{\ell} }{\rho _{\ell}-1} \right) T + \frac{3 a^2+6 a \sigma ^2+6 \eta _0 \sigma ^2+2 \sigma ^4}{6 \sqrt{2 \pi } \sigma }T^{3/2} \nonumber \\
 &+  \left(\vphantom{\sum_{1}^{N}} \frac{\left(2 a+\sigma ^2\right) \left(2 a+4 \eta _0+\sigma ^2\right)+4 \eta _1 \sigma ^2}{16}\right.\nonumber\\
 &\left.\phantom{\frac{1}{16}}+\sum_{\ell=1}^{N}\frac{a_{\ell} \left(\left(\rho_{\ell}-1\right) \left(2 \omega _{\ell,0}+2 a \rho_{\ell}+\rho_{\ell}^2 \sigma ^2\right)+a_{\ell} \left(2 \rho_{\ell}-1\right)\right)}{2 \left(\rho_{\ell}-1\right){}^2}\right)T^2 + O(T^{5/2})
\end{align} 
where $\{\eta_{i}\}_{i\geq 0}$ and $\{\omega_{\ell,i}\}_{1 \leq N,\,i \geq 0}$ are defined in terms of the parameters of the process in \eqref{eq:eta} and \eqref{eq:omega} respectively. Finally we compose \eqref{eq:inv_series_impv} with \eqref{eq:exp_series} (for validity, computational formulas and convergence of series composition see Theorem 1.9b and 2.4d in \cite{cca}) to get, after dividing through by $T^{1/2}$,
\begin{align}\label{eq:sig_exp}
\hat{\sigma}(T) = \sigma + \sqrt{2 \pi } b_2 T^{1/2} +  \left(\frac{\pi ^{3/2} b_1^3}{6 \sqrt{2}}+\sqrt{2 \pi } b_3\right) T + \left(\frac{\pi ^{3/2} b_2 b_1^2}{2 \sqrt{2}}+\sqrt{2 \pi } b_4\right) T^{3/2} + O(T^{2}),
\end{align}
where $b_i$ is the coefficient of $T^{i/2}$ in \eqref{eq:exp_series} and \eqref{eq:sig_exp} is valid for $T$ small enough. If we assume no Gaussian component $\sigma$ in the underlying L\'{e}vy process $X$ then we get
\begin{align}\label{eq:imp_n_g}
\hat{\sigma}(T) = \sqrt{2 \pi } b_1 T^{1/2} + \sqrt{2 \pi } b_2 T^{3/2} +  \left(\frac{\pi ^{3/2} b_1^3}{6 \sqrt{2}}+\sqrt{2 \pi } b_3\right) T^{5/2}+ O(T^{7/2}),
\end{align}
where $b_i$ is the coefficient of $T^i$ in the expansion
\begin{align}\label{eq:imp_no_gauss}
C_X(T)&= \left(\delta_1 + \sum_{\ell = 1}^{N}\frac{a_{\ell}}{\rho_{\ell} - 1}\right)T + \left(\delta_2 + \sum_{\ell=1}^{N}\frac{a_{\ell} \left(2 \left(\rho_{\ell}-1\right) \left(\omega_{\ell,0}+a \rho_{\ell}\right)+a_{\ell} \left(2 \rho_{\ell}-1\right)\right)}{2 \left(\rho_{\ell}-1\right){}^2}\right)T^{2}\nonumber\\
& + \left(\delta_3 +\sum _{\ell =1}^N \frac{a_{\ell }}{6 \left(\rho _{\ell }-1\right){}^3}
\left [\vphantom{\sum _{\ell =1}^N}3 a_{\ell } \left(\rho _{\ell }-1\right) \left(\rho _{\ell } \left(\rho _{\ell } \left(a-\omega _{\ell ,1}\right)+2 \omega _{\ell ,0}+\omega _{\ell ,1}\right)-\omega _{\ell ,0}\right) \right.\right. \nonumber \\ & \left. \left.  \hphantom{( \delta_3 + + } + 3 \left(\rho _{\ell }-1\right){}^2 \left(\omega _{\ell ,0}+a \rho _{\ell }\right){}^2+a_{\ell }^2 \left(3 \left(\rho _{\ell }-1\right) \rho _{\ell }+1\right)\vphantom{\sum _{\ell =1}^N}\right]
\right)T^{3} + O(T^{4}),
\end{align}
where 
\begin{align*}
\delta_1 =a ,\quad \delta_2 = \frac{a \left(a+2 \eta _0\right)}{2},\quad\text{ and }\quad \delta_3 = \frac{1}{6} a \left(a^2+3 a \eta _1+3 \eta _0 \left(a+\eta _0\right)\right),
\end{align*}
when $a > 0$ and zero otherwise. In general $\delta_i$ is equal to the $(i+1)$-th term of the expansion of $\xi_M\varsigma_M\zeta'_{M}$ divided by $i!$ when $a > 0$ and zero otherwise (see Lemma \ref{lem:at_the_money} and Theorem \ref{theo:at_the_money}).
In both cases we can show agreement with the general one-term results found in Proposition 5 of \cite{tankcar}. These show that for processes with finite second moment and Gaussian components we have $\lim_{T\rightarrow 0}\hat{\sigma}(T) = \sigma$ (compare Formula \ref{eq:sig_exp}) and for finite variation processes the one-term asymptotic expansion is
\begin{align*}
\hat{\sigma}(t) \sim \sqrt{2\pi}\max\left\{\int{(e^{x} - 1)^+}\nu(\d x),\,\int{(1-e^{x})^+}\nu(\d x)\right\}T^{1/2},
\end{align*}
where $\nu(\d x)$ is the L\'{e}vy measure. It is easy to confirm that in our case 
\begin{align*}
\int(e^x - 1)\nu(\d x) = \sum_{\ell = 1}^{N}\frac{a_{\ell}}{\rho_{\ell} - 1}\quad\text{ and }\quad \int(1-e^x)\nu(\d x) = \sum_{\ell = 1}^{\hat{N}}\frac{\hat{a}_{\ell}}{\hat{\rho}_{\ell} + 1},
\end{align*}
and that the risk neutral condition implies $a = \sum_{\ell = 1}^{\hat{N}}\frac{\hat{a}_{\ell}}{\hat{\rho}_{\ell} + 1} - \sum_{\ell = 1}^{N}\frac{a_{\ell}}{\rho_{\ell} - 1}$. In other words, $b_1 = \delta_1 + \sum_{\ell = 1}^{N}\frac{a_{\ell}}{\rho_{\ell} - 1} = \max\left\{\int{(e^{x} - 1)^+}\nu(\d x),\,\int{(1-e^{x})^+}\nu(\d x)\right\}$.\\ \\
\noindent In principle, this technique gives us a method for computing the full short-time asymptotic expansion of the at-the-money implied volatility in terms of the original parameters of the process. This should be compared to the two term expansion of \cite{figruot}, which, to the best of the author's knowledge, is the best result to date for exponential L\'{e}vy models -- albeit in a more general setting. Practically, we see that the formulas quickly become large, and symbolic computation is time consuming. We have computed $\hat{\sigma}(T)$ symbolically up to 6 terms; the result up to six terms can be obtained by employing the software package on the author's website.\\ \\ 
\noindent Note that if computation of $\hat{\sigma}(T)$ is the goal, i.e. if we fix numeric values for the parameters in advance, then we can easily compute one hundred or more terms. In Figure \ref{fig:imp_vol_1} we compute 1, 2, 5, 10, and 100-term expansions of $\hat{\sigma}(T)$ for the process defined by Parameter Set 1, $\sigma = 0.042$ and $a=0.111875$ and also the errors $\vert C_B(\hat{\sigma}(T),T) - C_X(T)\vert$. In Figure \ref{fig:imp_vol_2} we do the same for the process defined by Parameter Set 2, $\sigma = 0$ and $a=0.103896$. The markers blue $\circ$, red $\square$, purple $\bigtriangleup$, green $\bigtriangledown$, and magenta $\diamond$ represent the 1, 2, 5, 10, and 100-term expansions respectively. We find it takes approximately 0.2 seconds to compute the 50 values of the 10-term approximation depicted in the figures and approximately 20 seconds for the 100-term approximation.\demo
\begin{figure}[!t]
\centering
\subfloat[]
{\label{fig:imp_vol_1_exp}\includegraphics[scale=0.48]{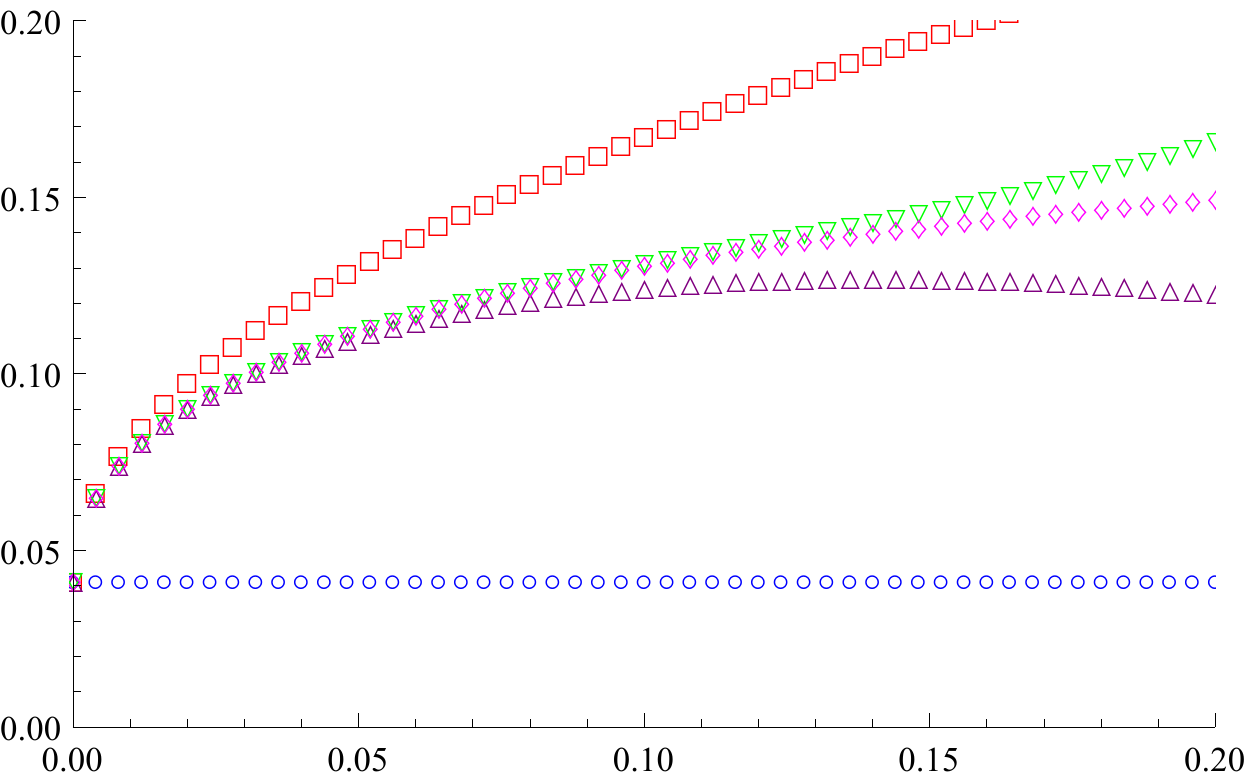}}
\subfloat[]
{\label{fig:imp_vol_1_err}\includegraphics[scale=0.48]{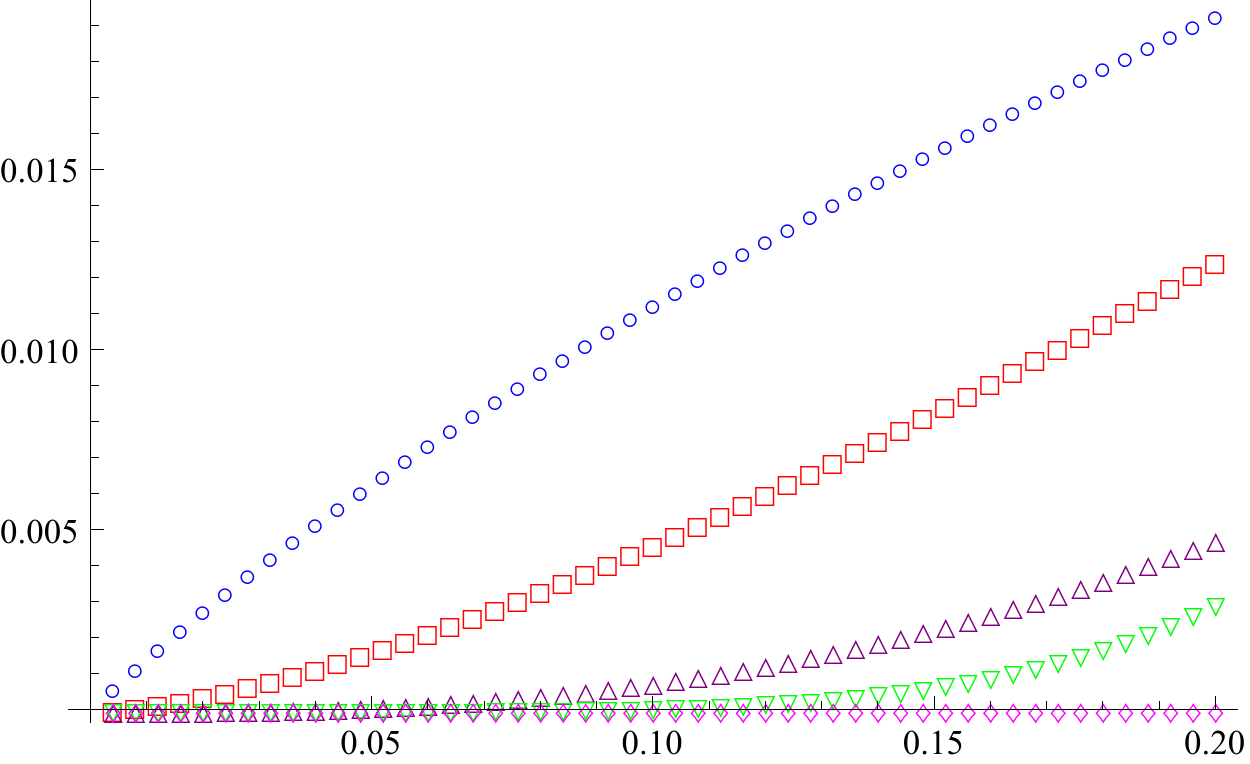}}
\subfloat[]
{\label{fig:imp_vol_1_zerr}\includegraphics[scale=0.48]{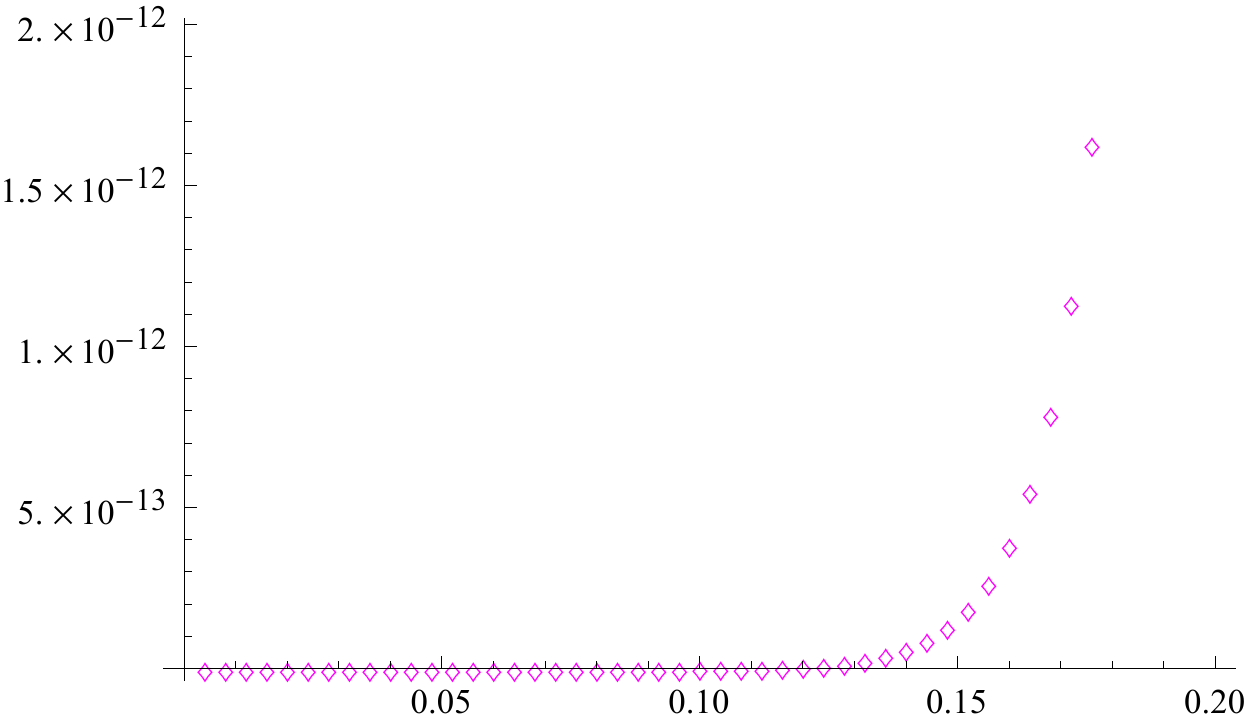}}
\caption{$\hat{\sigma}(T)$ for a process with Parameter Set 1, $\sigma = 0.042$, and $a=0.111875$ in \ref{fig:imp_vol_1_exp}. Errors for all expansions in Figure \ref{fig:imp_vol_1_err} and error of the 100-term expansion only in \ref{fig:imp_vol_1_zerr}.}
\label{fig:imp_vol_1}
\end{figure}
\begin{figure}[!t]
\centering
\subfloat[]
{\label{fig:imp_vol_2_exp}\includegraphics[scale=0.48]{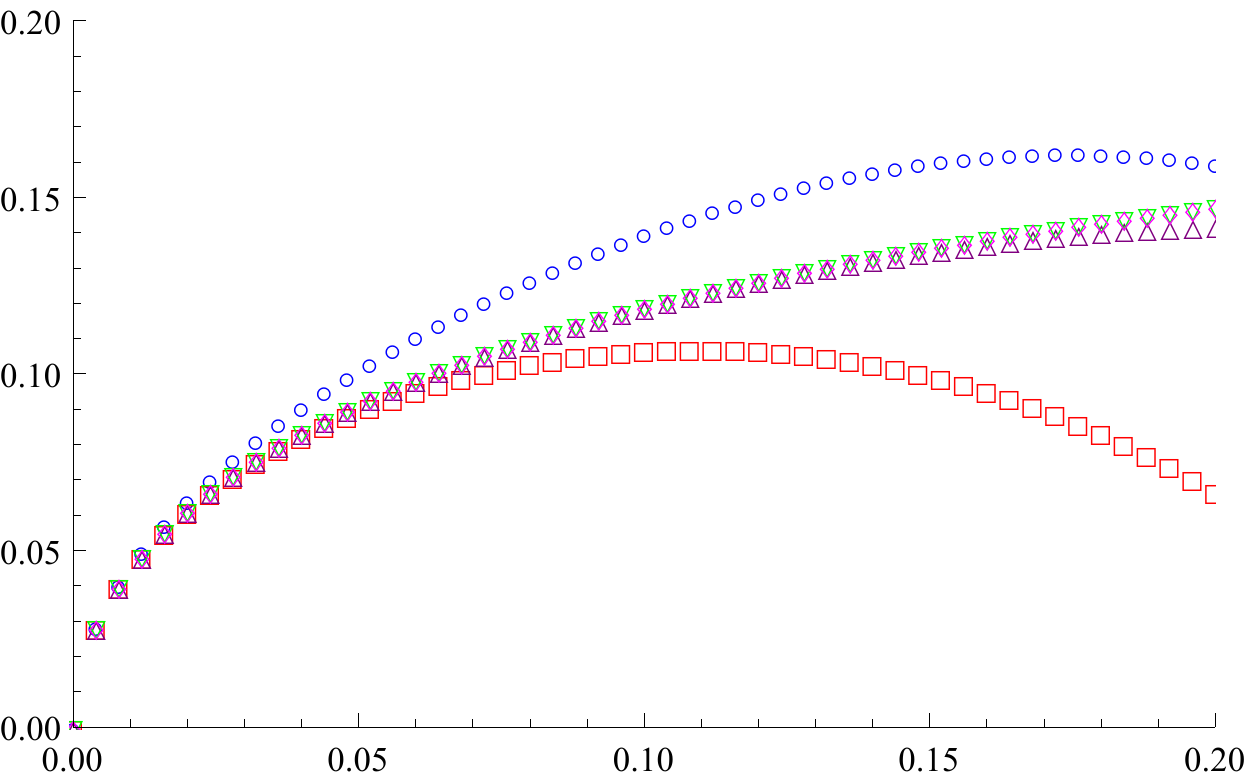}}
\subfloat[]
{\label{fig:imp_vol_2_err}\includegraphics[scale=0.48]{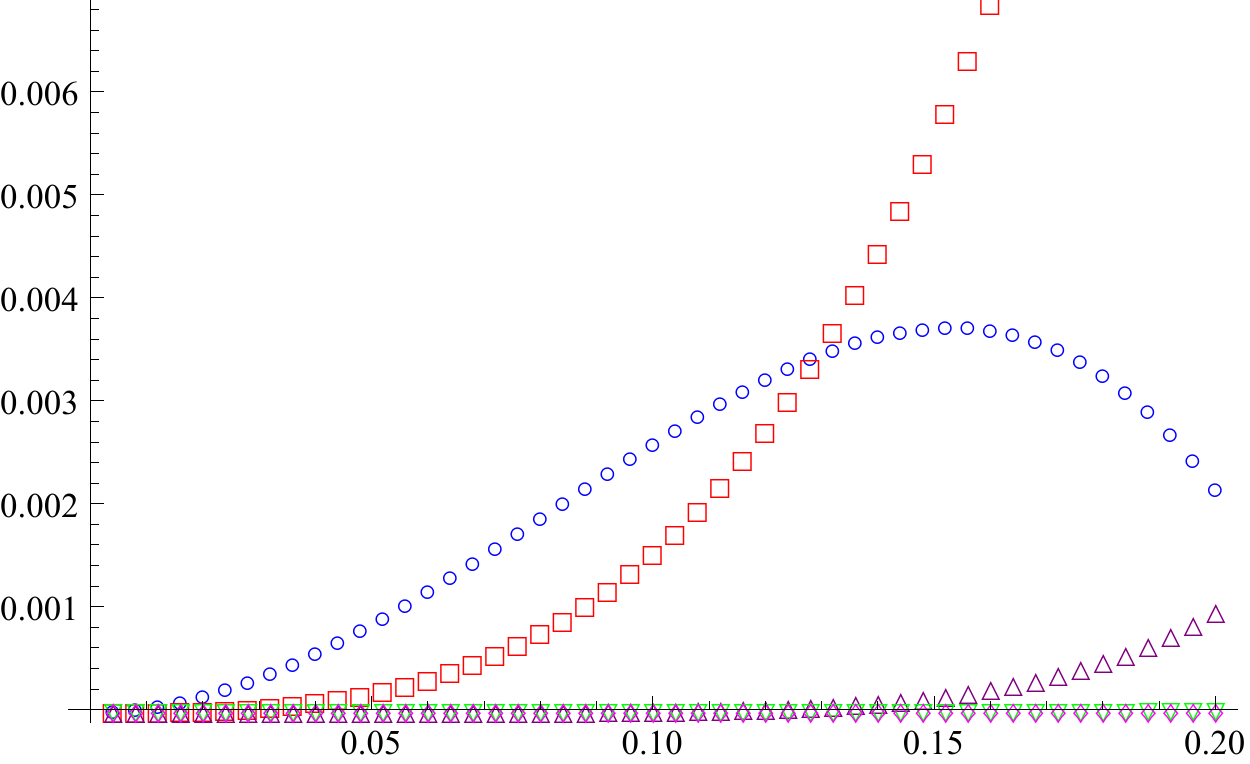}}
\subfloat[]
{\label{fig:imp_vol_2_zerr}\includegraphics[scale=0.48]{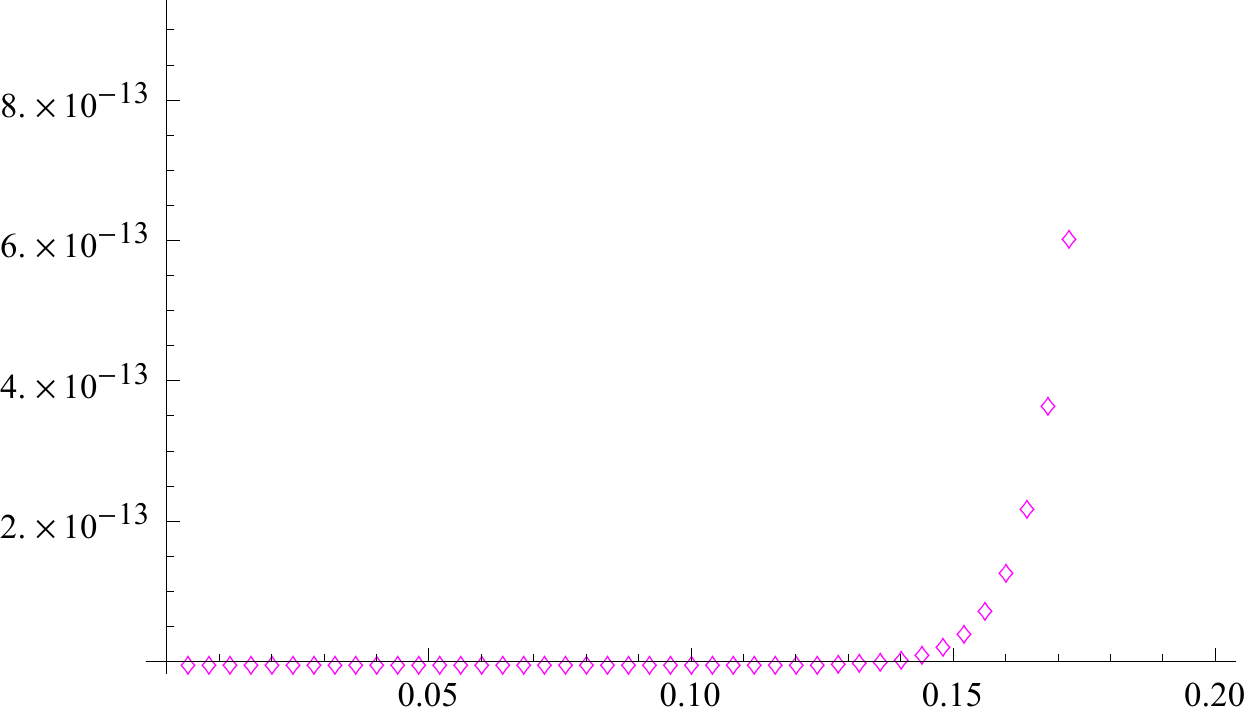}}
\caption{$\hat{\sigma}(T)$ for a process with Parameter Set 2, $\sigma = 0$ and $a=0.103896$ in \ref{fig:imp_vol_2_exp}. Errors for all expansions in Figure \ref{fig:imp_vol_2_err} and error of the 100-term expansion only in \ref{fig:imp_vol_2_zerr}.}
\label{fig:imp_vol_2}
\end{figure}
\end{example}
\noindent We may work out results for the in the money (ITM) and OTM cases, with the additional complication that we will have series with summands of the form $\frac{k^{-Dq^{1/2}}}{q^{n/2}}$ and $\frac{k^{-Dq}}{q^n}$ for some constant $D$. Before proceeding, let us consider these functions as Laplace transforms and establish some of their properties.
\begin{lemma}\label{lem:extra_laplace}
If $k > 1$ and $D > 0$ or $k < 1$ and $D < 0$ then
\begin{enumerate}[(i)]
\item 
\begin{align*}
\mathcal{L}\left\{ \frac{(t-c)^{n-1}}{(n-1)!}\ind(t \geq c)\right\}(q) = \frac{k^{-Dq}}{q^n}
\end{align*}
\item and
\begin{align*}
\mathcal{L}\left\{\frac{t^{-3/2}}{2\sqrt{\pi}(n-1)!} \int_c^{\infty} e^{-\tau^2/(4t)}\tau(\tau-c)^{n-1}\d \tau\right\}(q) = \frac{k^{-Dq^{1/2}}}{q^{n/2}}
\end{align*}
\end{enumerate}
where $n \in \mathbb{N}$ and $c = D\log(k)$.
\end{lemma}
\begin{proof}
Both (i) and (ii) can be proven using tables of Laplace transforms. For (i) see for example Entry 26 on pg. 339 in \cite{doetsch} together with the general rule for inverting $e^{-cs}G(s)$ where $G(s)$ has a known inverse (see for example pg. 337 in \cite{doetsch}). For (ii) we use result from (i) together with the rule for inverting $G(s^{1/2})$ where $G(s)$ has a known inverse; see again pg. 337 in \cite{doetsch}.
\end{proof}
\noindent Going forward let 
\begin{align}\label{eq:varphi}
\varphi_n(t;c) := \frac{t^{-3/2}}{2\sqrt{\pi}(n+1)!} \int_c^{\infty} e^{-\tau^2/(4t)}\tau(\tau-c)^{n+1}\d \tau,\quad t > 0,\, c \geq 0,\,n \in \mathbb{Z}^+.
\end{align}
\begin{lemma}\label{lem:extra_bnd}
The following inequality holds
\begin{align}\label{eq:eps_bnd}
\varphi_n(t;c) \leq \frac{t^{n/2}}{\Gamma\left(1 + \frac{n}{2}\right)},\quad t > 0,\, c \geq  0,\,n \in \mathbb{Z}^+,
\end{align}
with equality when $c = 0$.
It follows that $\lim_{t\rightarrow 0+}\varphi_n(t;c) = 0$ for $n \in \mathbb{N}$, $c \geq 0$.
\end{lemma}
\begin{proof}
Clearly $\varphi_{n}(t;c) \leq \varphi_{n}(t;0)$, and the latter can be evaluated explicitly in terms of the gamma function as
\begin{align}\label{eq:cool_gam}
\varphi_{n}(t;0) = \frac{2^{n+1}\Gamma\left(\frac{n+2}{2} + \frac{1}{2}\right)}{\sqrt{\pi}\Gamma(n+2)}t^{n/2}.
\end{align}
Using the well-known duplication formula for the gamma function, i.e. $\Gamma(z)\Gamma(z + 1/2) = 2^{1-2z}\sqrt{\pi}\Gamma(2z)$, shows that the right- hand side of \eqref{eq:cool_gam} is in fact equal to $t^{n/2}/\Gamma(n/2 + 1)$.
\end{proof}
\noindent In the remainder of the paper we interpret the notation $\varphi_n(0;c)$ as $\lim_{t\rightarrow0+}\varphi_n(t;c) = 0$.
\begin{corollary}\label{cor:complx_bnd}
Let $z = x + iy \in K \subset \mathbb{H}$, where $K$ is a compact set, and define $u := \vert z \vert^2/x$. Then, there exists $M > 0$ such that 
\begin{align*}
\vert \varphi_n(z;c) \vert < M\frac{u^{n/2}}{\Gamma\left(1 + \frac{n}{2}\right)},\quad z\in K,\,c \geq 0,\,n \in \mathbb{Z}^+.
\end{align*}
\end{corollary}
\begin{proof}
\begin{align*}
\vert \varphi_n(z,c) \vert \leq \frac{\vert z \vert ^{-3/2}}{2\sqrt{\pi}(n+1)!} \int_c^{\infty}\tau \vert e^{-\tau^2/(4z)}\vert (\tau-c)^{n+1}\d \tau = \left(\frac{u}{\vert z \vert}\right)^{3/2}\varphi_n(u;c) <\left(\frac{\vert z \vert}{x}\right)^{3/2}\frac{u^{n/2}}{ \Gamma\left(1 + \frac{n}{2}\right)},
\end{align*}
where the last inequality is due to Lemma \ref{lem:extra_bnd}. Since $K$ is compact and $z \mapsto \left(\vert z \vert/x\right)^{3/2}$ is a continuous function on $\mathbb{H}$ we know that its supremum is attained and finite on $K$. This gives the result.
\end{proof}
\noindent For computation of the functions $\{\varphi_n(t;c)\}_{n\geq 0}$ we introduce the $\{\text{Hh}_n(z)\}_{n\geq -1}$ functions, which are defined as follows:
\begin{align}\label{eq:Hh_def}
 \text{Hh}_n(z) := \frac{1}{n!}\int_{z}^{\infty}e^{-w^2/2}(w - z)^{n}\d w,\quad n \in \mathbb{Z}^+,\, z \in \c,
\end{align}
and $\text{Hh}_{-1}(z) := e^{z^2/2}$. We observe that $\text{Hh}_0(z) = (\sqrt{\pi}/\sqrt{2})\text{erfc}(-z/\sqrt{2})$ where \\ $\text{erfc}(z) =(2/\sqrt{\pi}) \int_{z}^{\infty}e^{-w^2/2}\d w$ is the complementary error function. Conveniently we have the three-term recurrence
\begin{align}\label{eq:h_recurs}
\text{Hh}_n(z) = \frac{1}{n}\text{Hh}_{n-2}(z) - \frac{x}{n}\text{Hh}_{n-1}(z),\quad n \in \mathbb{N}
\end{align}
and also the simple recurrence
\begin{align}\label{eq:hh_deriv}
\frac{\d }{\d z} \text{Hh}_n(z) = - \text{Hh}_{n-1}(z);
\end{align}
see Formulas 7.2.5, 7.2.9, 7.2.10, and 19.14 in \cite{AS}
\begin{lemma}\label{lem:extra_analytic}
For fixed $c \geq 0$, functions $\{\varphi_n(t;c)\}_{n\geq 0}$ can be analytically continued to $z \in \c^+$ and we have
\begin{align}\label{eq:phi_and_h}
\varphi_n(z;c) = \frac{2^{(n+1)/2}z^{n/2}}{\sqrt{\pi}}\textnormal{Hh}_n(c/\sqrt{2z}),\quad z \in \mathbb{H},\, c \geq 0,\,n \in \mathbb{Z}^+.
\end{align} 
\end{lemma}
\begin{proof}
First we use the dominated convergence theorem to show that differentiation on the integral is permissible in \eqref{eq:varphi} provided $z \in \mathbb{H}$, i.e. that $\varphi_n(z;c)$ is analytic on $\mathbb{H}$. Restricting $\textnormal{Hh}_n(z)$ to real values $z >0$, performing one iteration of integration by parts in \eqref{eq:Hh_def}, and then employing the substitution $w = \tau/(\sqrt{2t})$ establishes the identity \eqref{eq:phi_and_h} for real $z > 0$. Since $\textnormal{Hh}_n(z)$ is an entire function the claim follows by an argument of analytic continuation.
\end{proof}
\begin{remark}
In \cite{kou_sad} the formulas derived for European call option prices under the double exponential model are given in terms of series of integrals of the $\{\text{Hh}_n(x)\}_{n\geq -1}$ functions. It is therefore not surprising that they appear here; however, we will see that we do not need to integrate further, we can express our formulas for call and put options in terms of the functions $\{\varphi_n(t;c)\}_{n\geq 0}$ directly. \rdemo
\end{remark}
\noindent Now we are ready to proceed with the OTM case for call and put options; note that once these are established the ITM price is then easily computed via the put-call parity. As alluded to, things are a little more complicated due to the fact that $F(q)$ and $W(q)$ involve terms of the form $\{k^{-\zeta_{\ell}}\}_{1\leq \ell \leq M}$ and $\{k^{\hat{\zeta}_\ell}\}_{1\leq\ell\leq \hat{M}}$. To keep the notation as simple as possible, let us define
\begin{align}\label{eq:beta_1}
\beta_{\ell} := k^{-\zeta_\ell},\quad 1 \leq \ell \leq N,\quad\quad \hat{\beta}_{\ell} := k^{\hat{\zeta}\ell}\quad 1\leq \ell \leq \hat{N}.
\end{align}
Further, for the case $\sigma > 0$ we define
\begin{align}\label{eq:beta_2}
\beta_{M} := k^{(2q/\sigma^2)^{1/2}}k^{-\zeta_M}\quad\text{ and }\quad \hat{\beta}_{\hat{M}}:= k^{-(2q/\sigma^2)^{1/2}}k^{\hat{\zeta}_{\hat{M}}},
\end{align}
and when $\sigma= 0$ and $a > 0$ (resp. $a < 0$) we set
\begin{align*}
\beta_M := k^{q/a}k^{-\zeta_{M}}\quad \left(\text{ resp. }\hat{\beta}_{\hat{M}} := k^{q/a}k^{\hat{\zeta}_{\hat{M}}}\right).
\end{align*}
Finally, for the following Lemma \ref{lem:in_the_money}, Theorem \ref{theo:in_the_money} and Theorem \ref{theo:out_the_money} we define for each $n \in \mathbb{Z}^+$,
\begin{align}\label{eq:coeffs}
b_n := s_n(m_{1,n},\,m_{2,n},\,\ldots,\,m_{N,n}),\quad c_n := m_{M,n},\quad m_{\ell,n} := m_{n}(\bar{\xi}_{\ell,n},\,\bar{\varsigma}_{\ell,n},\,\bar{\beta}_{\ell,n},\,\bar{\zeta'}_{\ell,n}), \quad 1 \leq \ell \leq M,
\end{align}
and
\begin{align*}
\hat{b}_n := s_n(\hat{m}_{1,n},\,\hat{m}_{2,n},\,\ldots,\,\hat{m}_{\hat{N},n}),\quad \hat{c}_n := \hat{m}_{\hat{M},n},\quad \hat{m}_{\ell,n} := m_{n}(\bar{\hat{\xi}}_{\ell,n},\,\bar{\hat{\varsigma}}_{\ell,n},\,\bar{\hat{\beta}}_{\ell,n},\,\overline{\hat{\zeta'}}_{\ell,n}), \quad 1 \leq \ell \leq \hat{M},
\end{align*}
\begin{lemma}\label{lem:in_the_money}
Suppose that $X$ fulfils the risk neutral condition and $S_0 < K$.
\begin{enumerate}[(i)]
\item If $\sigma > 0$ then there exists $R > 0$ such that the following equality holds for $q \in \c^+_{R}$ and the series converge for $q \in \c_R$:
\begin{align*}
F(q) = k\sum_{n=2}^{\infty}\frac{b_n}{q^{n}} + k\sum_{n=3}^{\infty}\frac{c_nk^{-(2q/\sigma^2)^{1/2}}}{q^{n/2}}.
\end{align*}
\item If $\sigma = 0$ and $a > 0$ then there exists $R > 0$ such that the following equality holds and the series converge for $q \in \mathbb{C}_R$:
\begin{align*}
F(q) = k\sum_{n=2}^{\infty}\frac{b_n}{q^{n}} + k\sum_{n=2}^{\infty}\frac{c_nk^{-q/a}}{q^n}.
\end{align*}
\item If $\sigma = 0$ and $a \leq 0$ then there exists $R > 0$ such that the following equality holds and the series converges for $q \in \mathbb{C}_R$:
\begin{align*}
F(q) = k\sum_{n=2}^{\infty}\frac{b_n}{q^{n}}.
\end{align*}
\end{enumerate}
\end{lemma}
\begin{proof}
Use Formula \ref{eq:medium} and apply Corollaries \ref{cor:with_sig}, \ref{cor:without_sig}, and \ref{cor:norm}.
\end{proof}
\begin{theorem}[\textbf{OTM call option price}]\label{theo:in_the_money}
Suppose that $X$ fulfils the risk neutral condition and $S_0 < K$.
\begin{enumerate}[(i)]
\item Let $\sigma > 0$, then
\begin{align*}
C(T) = e^{-r T}K  \left(\sum_{n=1}^{\infty}b_{n+1}\frac{T^n}{n!} + \sum_{n=1}^{\infty}c_{n+2}\varphi_n\left(T;c\right)\right),\quad T \geq 0,
\end{align*}
where $c = \frac{\sqrt{2}\log(k)}{\sigma}$. The first series on the right-hand side converges for $T \in \c$ and the second converges for $T \in \mathbb{H}$. The second series converges uniformly on compact subsets of $\mathbb{H}$ and may be differentiated termwise.
\item Let $\sigma = 0$ and $a >0$, then
\begin{align*}
C(T) = e^{-r T}K\left(\sum_{n=1}^{\infty}b_{n+1}\frac{T^{n}}{n!} + \ind(T \geq c)\sum_{n=1}^{\infty}c_{n+1}\frac{(T-c)^{n}}{n!}\right),\quad T \geq 0,
\end{align*}
where $c=\frac{\log(k)}{a}$. Each series on the right-hand side converges for $T\in\c$.
\item Let $\sigma = 0$ and $a \leq 0$, then
\begin{align*}
C(T) = e^{-r T}K\sum_{n=1}^{\infty}b_{n+1}\frac{T^{n}}{n!},\quad T \geq 0,
\end{align*}
where the series on the right-hand side converges for $T \in \c$.
\end{enumerate}
\end{theorem}
\begin{proof}
We prove only (i) as (ii) and (iii) are derived in similar fashion. First, let us consider the series
\begin{align*}
H_2(q) := \sum_{n=3}^{\infty}\frac{c_nk^{-(2q/\sigma^2)^{1/2}}}{q^{n/2}}.
\end{align*}
We will proceed in three steps. \textbf{Step 1:}
By employing Theorem \ref{theo:tbt} we may conclude that
\begin{align}\label{eq:comp_lap}
\mathcal{L}\left\{\sum_{n=1}^{\infty}c_{n+2}\frac{t^{n/2}}{\Gamma\left(\frac{n}{2} + 1\right)}\right\}(q) = k^{(2q/\sigma^2)^{1/2}}H_2(q)
\end{align}
and that the series $s(t) :=\sum_{n=1}^{\infty}c_{n+2}\frac{t^{n/2}}{\Gamma\left(\frac{n}{2} + 1\right)}$  converges for almost all $t \geq 0$. This, however, implies that it converges for all $t \in \c$. \textbf{Step 2:} Employing Theorem \ref{theo:tbt} again together with Lemma \ref{lem:extra_laplace} we find that
\begin{align}\label{eq:base_series} \mathcal{L}\left\{\sum_{n=1}^{\infty}c_{n+2}\varphi_n\left(t;c\right)\right\}(q) = H_2(q)
\end{align}
and that the convergence of the series $h_2(t) := \sum_{n=1}^{\infty}c_{n+2}\varphi_n\left(t;c\right)$ is absolute for almost all $t \geq 0$. However, via Lemma \ref{lem:extra_bnd}, the fact that $s(t)$ converges for all $t \in \c$, and the comparison test for series, it is clear that absolute convergence of $h_2(t)$ holds for all $t \geq 0$. \textbf{Step 3:} To show analyticity we use Weierstrass' $M$-test and Theorem on Uniformly Convergent Series of Analytic Functions (see Theorems 15.2 and 15.6 in \cite{mark}). Together these tell us that if for every compact set $K \subset \mathbb{H}$ we can find a sequence $\{M_n\}_{n\geq 1}$ such that $\vert c_{n+2}\varphi_n(z;c) \vert < M_n$, for $z =x + iy \in K,\,n \in \mathbb{N}$, and $\sum_{n=1}^{\infty}M_n < \infty$, then $h_2(t)$: a) is an analytic function on $\mathbb{H}$; b) converges uniformly on compact subsets of $\mathbb{H}$; and c) can be differentiated term-by-term and the resulting series is again uniformly convergent on compact subsets of $\mathbb{H}$. From Corollary \ref{cor:complx_bnd} we know that for every compact $K \subset \mathbb{H}$
\begin{align*}
\vert c_{n+2}\varphi_n(z;c) \vert \leq M \vert c_{n+2}\vert\frac{u^{n/2}}{\Gamma(1+\frac{n}{2})},\quad z \in K,\,n \in \mathbb{Z}^+,
\end{align*}
where $u := u(z) := \vert z \vert^2/x$ and $M$ is some positive constant. By continuity of $u$ and compactness of $K$ the function $u(z)$ must attain at maximum $0 \leq S < \infty$ on $K$. From our discussion, it is clear that the series $s(t)$ converges absolutely at $S$, and so, if we take
\begin{align*}
M_n = M\vert c_{n+2} \vert \frac{S^{n/2}}{\Gamma(1+\frac{n}{2})},
\end{align*}
then the above criteria are met and analyticity, termwise differentiation, and uniform convergence follow. Now we may repeat Step 1 for the series $H_1(q) := \sum_{n=2}^{\infty}\frac{b_n}{q^2}$ to show that
\begin{align*}
\mathcal{L}\left\{\sum_{n=1}^{\infty}b_{n+1}\frac{t^n}{n!}\right\}(q) = H_1(q)
\end{align*}
and $h_1(t) := \sum_{n=1}^{\infty}b_{n+1}\frac{t^n}{n!}$ converges for all $t \in \c$. Then $k(h_1(t) + h_2(t))$ defines a continuous function whose Laplace transform $k(H_1(q) + H_2(q))$ is equal to the Laplace transform of the continuous function $f(t)$. It follows that $f(t) = k(h_1(t) + h_2(t))$, $t \geq 0$ and the result follows.
\end{proof}
\noindent From Formula \ref{eq:hard} it is clear that we may develop an analogous expansion of the OTM put option price (i.e. when $k<1$) using exactly the same methods in Lemma \ref{lem:in_the_money} and Theorem \ref{theo:in_the_money}. Rather than developing this in detail, we simply gather the results in the following theorem.
\begin{theorem}[\textbf{OTM put option price}]\label{theo:out_the_money}
Suppose that $X$ fulfils the risk neutral condition and $S_0 > K$.
\begin{enumerate}[(i)]
\item Let $\sigma > 0$, then
\begin{align*}
P(T) = e^{-r T}K  \left(\sum_{n=1}^{\infty}\hat{b}_{n+1}\frac{T^n}{n!} + \sum_{n=1}^{\infty}\hat{c}_{n+2}\varphi_n\left(T;c\right)\right),\quad T \geq 0,
\end{align*}
where $c = -\frac{\sqrt{2}\log(k)}{\sigma}$. The first series on the right-hand side converges for $T \in \c$ and the second converges for $T \in \mathbb{H}$. The second series converges uniformly on compact subsets of $\mathbb{H}$ and may be differentiated termwise.
\item Let $\sigma = 0$, $a < 0$, then
\begin{align*}
P(T) = e^{-r T}K\left(\sum_{n=1}^{\infty}\hat{b}_{n+1}\frac{T^{n}}{n!} +\ind(T \geq c)\sum_{n=1}^{\infty}\hat{c}_{n+1}\frac{(T-c)^{n}}{n!}\right),\quad T \geq 0,
\end{align*}
where $c=\frac{\log(k)}{a}$. Each series on the right-hand side converges for $T\in\c$.
\item Let $\sigma = 0$, $a \geq 0$, then
\begin{align*}
P(T) = e^{-r T}K\sum_{n=1}^{\infty}\hat{b}_{n+1}\frac{T^{n}}{n!},\quad T \geq 0,
\end{align*}
where the series on the right-hand side converges for $T \in \c$.
\end{enumerate}
\end{theorem}
\begin{example}[\textbf{Computation of option prices}]
We can use the results of Theorems \ref{theo:at_the_money}, \ref{theo:in_the_money}, and \ref{theo:out_the_money} to compute European call and put option prices. For all the cases, e.g. ITM call options, not covered by these theorems explicitly, we can use the put-call parity relationship. To discuss our computational approach, let us assume that we wish to compute an OTM call option price for a process $X$ with $\sigma > 0$ so that the formula of Theorem \ref{theo:out_the_money} (i) applies. Retracing our steps to Lemma \ref{lem:in_the_money}, and recalling the definition \eqref{eq:coeffs} of the coefficients $\{b_n\}_{n\geq 2}$ and $\{m_{\ell,n}\}_{n\geq 0,\,1 \leq \ell \leq M}$ we see that the first term in our computation of the price is of the form
\begin{align*}
\sum_{n=1}^{\infty}b_{n+1}\frac{T^n}{n!} = \sum_{\ell=1}^{N}\sum_{n=1}^{\infty}m_{\ell,n+1}\frac{T^n}{n!},
\end{align*}
i.e. a sum of $N$ series. Since we assumed $\sigma > 0$ we also need to compute the extra series
\begin{align*}
\sum_{n=1}^{\infty}c_{n+2}\varphi_n\left(T;c\right).
\end{align*}
In total then, we need to compute the sum of $M = N+1$ series. Since these series may converge at different rates it makes sense to compute/truncate them individually so as not to expend more effort than is necessary. Going forward let $\mathbf{M}$ (resp.$\hat{\mathbf{M}}$) denote the vector in $\mathbb{N}^{M}$ (resp. $\mathbb{N}^{\hat{M}}$) containing the points of truncation of our series. \\ \\
\noindent We compute a ITM call option price for a process with Parameter Set 1, $\sigma = 0.042$, $a = 0.141875$ and option parameters $K = 90$, $S_0 = 95$, and $r = 0.03$. The results are shown in Table \ref{tab:opt_calc}. In Tables \ref{tab:opt_calc_atm} and \ref{tab:opt_calc_otm} ATM and OTM call option prices are shown for a process with Parameter Set 2, $\sigma = 0$, $a = 0.133896$ and option parameters $S_0=K=300$ and $S_0 = 10$ and $K=11$ respectively. \\ \\
\noindent As basis of comparison, we also include a price computed via the numerical Fourier inversion technique of \cite{Carr}. Specifically, we consider the price $C(T)$ as a function not of $T$ but of $s := \log(K)$ the log-strike. For the purpose of this example only we will write $C_T(s)$ to denote the call price.  Then, for fixed $T$ it is easy to show, that $\mathcal{L}\{e^{rT}C_T(s)\}(z) = \frac{S_0^{(1-z)}e^{T\psi(1-z)}}{z(z-1)}$, $1 - \rho_1 < \Re(z) < 0$. Thus we can compute the price by evaluating
\begin{align}\label{eq:carr_laplace}
C_T(s) = \frac{e^{-rT}}{2\pi i}\int_{c + i\r}\frac{S_0^{(1-z)}e^{T\psi(1-z)}}{z(z-1)}\d z = e^{-rT}\times \Re \left(\frac{e^{cs}}{\pi}\int_{0}^{\infty}e^{isu}\frac{S_0^{(1-c-iu)}e^{T\psi(1-c-iu)}}{(c+iu)(c+iu-1)}\d u \right),
\end{align}
where $1 -\rho_1 < c < 0$. We compute this integral using the aforementioned (Example \ref{ex:barr}) Filon quadrature, and by truncating at an upper limit of $10^5$. Depending on $T$ we use between $2\times 10^5$ and $4 \times 10^5$ discretization steps. 
 \\\ \\
\noindent We find that prices can generally be computed quickly and accurately using the analytical formulas for maturities less than 1. For larger maturities, the series may converge too slowly (see for example the $T=0.9$ case of Table \ref{tab:opt_calc}). The major benefit of this approach, is that option prices for a range of short to medium length maturities can be computed very quickly. For example, we see from Table \ref{tab:opt_calc} that setting $\hat{\mathbf{M}} = (15, 15, 15, 15, 15, 30, 30, 60)$ gives accurate prices for expiries smaller than $0.5$. The time (including the time to calculate the coefficients of the series) to compute ITM call option prices for 100 different expiries in the interval $[0,0.5]$  is approximately three seconds. When $\sigma = 0$, i.e. when we do not need to compute the functions $\{\varphi_n(T;c)\}_{n\geq 1}$, the time is faster still. Setting $\mathbf{M} = (10, 10, 10, 10, 10, 12, 12, 16)$ in the example from Table \ref{tab:opt_calc_otm} we find we can accurately compute 100 prices with expiries in the interval $[0,0.5]$ in less than half of a second.\\ \\
\noindent By comparison, even if we store all common elements needed for numerical inversion of the Laplace transform for different expiries, our computation takes at least 16 seconds for 100 prices for either of the cases mentioned above. Notice that a fast Fourier approach is not applicable here, since the transform is in the variable $s$ and not $T$. If we wanted to use a fast Fourier approach, we could consider inverting the expressions \eqref{eq:medium} -- \eqref{eq:hard} which are transforms in $T$. This would come at the expense of computing the solutions $\{\zeta_{\ell}\}_{1 \leq \ell \leq M}$, $\{\hat{\zeta}_{\ell}\}_{1 \leq \ell \leq \hat{M}}$ for each step of the algorithm (c.f. Example \ref{ex:barr}). \demo
\renewcommand{\arraystretch}{1.5}
\begin{table}
\begin{center}
\begin{tabular}{ | c || c | c | c | c | c || c |}
  \hline 
\diagbox{$\hat{\mathbf{M}}$}{$T$} & $0.01$ & $0.1$ & $0.2$ & $0.5$ & $0.9$ & Time \\
\hline \hline 
$(2,2,2,2,2,4,4,8)$ & 5.09974 & 5.94683 & 6.78061 & 7.23382 &-24.48527 & 0.160\\
\hline
$(4, 4, 4, 4, 4, 6, 6, 10)$ & 5.09975 &  5.94753 & 6.79477 & 6.67534 & -158.31075 & 0.196  \\
\hline
$(6, 6, 6, 6, 6, 8, 8, 12)$ & 5.09975 & 5.94755 & 6.79746 & 7.95322 & -265.66525 & 0.244 \\
\hline
$(8, 8, 8, 8, 8, 10, 10, 14)$ & 5.09975 & 5.94755 & 6.79760 & 9.19484 & 155.77080 & 0.300  \\
\hline
$(10, 10, 10, 10, 10, 12, 12, 16)$ & 5.09975&  5.94755 & 6.79759 & 9.49692 & 1457.00312 & 0.340 \\
\hline
$(15, 15, 15, 15, 15, 30, 30, 60)$ &5.09975& 5.94755& 6.79759& 8.95421& -12.68581 & 1.892 \\
\hline
$(15, 15, 20, 20, 20, 40, 40, 100)$ & 5.09975 & 5.94755 & 6.79759 & 8.95421 & 11.32685 & 5.892 \\
\hline
$(15, 15, 25, 25, 35, 50, 50, 110)$ & 5.09975 & 5.94755 & 6.79759 & 8.95421 & 11.29891 & 7.692 \\
\hline
\hline
Fourier comparison & 5.09975 & 5.94755 & 6.79759 & 8.95421 & 11.29892 & 1.94 \\
\hline
\end{tabular}
\caption{Call option prices for the process with Parameter set 1, $\sigma =0.042$, $a=0.141875$. Option parameters $K=90$, $S_0=95$, $r=0.03$. Time is the time to compute the 5 option prices in seconds.}\label{tab:opt_calc}
\end{center}
\end{table}
\renewcommand{\arraystretch}{1.5}
\begin{table}
\begin{center}
\begin{tabular}{ | c || c | c | c | c | c || c |}
  \hline 
\diagbox{$\mathbf{M}$}{$T$} & $0.01$ & $0.1$ & $0.2$ & $0.5$ & $0.9$ & Time \\
\hline \hline 
$(2,2,2,2,2,4,4,8)$ & 0.61954 & 5.25306 & 9.07478 & \
-515.10360 & -72839.82457& 0.12\\
\hline
$(4, 4, 4, 4, 4, 6, 6, 10)$ & 0.61954 & 5.25119& 9.20011& \
-637.32154& -296414.11217& 0.148  \\
\hline
$(6, 6, 6, 6, 6, 8, 8, 12)$ & 0.61954 & 5.25121 & 9.23466 &\
-544.43333& -844397.61170 & 0.188 \\
\hline
$(8, 8, 8, 8, 8, 10, 10, 14)$ & 0.61954 & 5.25121 & 9.23939 & \
-339.58140 & -1,776,370.99729 & 0.232  \\
\hline
$(10, 10, 10, 10, 10, 12, 12, 16)$ & 0.61954 & 5.25121& 9.23987& \
-157.10396& -2,871,257.37102& 0.252 \\
\hline
$(15, 15, 15, 15, 15, 30, 30, 60)$ &0.61954 & 5.251214 & 9.23991 & \
18.14807 & 26.98182 & 1.304 \\
\hline
$(15, 15, 20, 20, 20, 40, 40, 100)$ & 0.61954 & 5.25121 & 9.23991 & \
18.14807 & 26.98185 & 4.076 \\
\hline
\hline
Fourier comparison & 0.61954 & 5.25121 &  9.23991 & 18.14807 & 26.98185 & 1.04\\
\hline
\end{tabular}
\caption{Call option prices for the process with Parameter set 2, $\sigma =0$, $a=0.133896$. Option parameters $S_0=K=300$, $r=0.03$. Time is the time to compute the 5 option prices in seconds.}\label{tab:opt_calc_atm}
\end{center}
\end{table}
\renewcommand{\arraystretch}{1.5}
\begin{table}
\begin{center}
\begin{tabular}{ | c || c | c | c | c | c || c |}
  \hline 
\diagbox{$\mathbf{M}$}{$T$} & $0.01$ & $0.1$ & $0.2$ & $0.5$ & $0.9$ & Time \\
\hline \hline 
$(2,2,2,2,2,4,4,8)$ & 0.00128 & 0.01532 & 0.03684 & \
0.14880 & 0.44918 & 0.18\\
\hline
$(4, 4, 4, 4, 4, 6, 6, 10)$ & 0.00128 & 0.01532 & 0.03678 & \
0.14508 & 0.38900 & 0.184  \\
\hline
$(6, 6, 6, 6, 6, 8, 8, 12)$ & 0.00128 & 0.01532 & 0.03678 & \
0.14486 & 0.38104 & 0.228 \\
\hline
$(8, 8, 8, 8, 8, 10, 10, 14)$ & 0.001282 & 0.01532 & 0.03678 & \
0.14489 & 0.38465 & 0.280  \\
\hline
$(10, 10, 10, 10, 10, 12, 12, 16)$ & 0.00128 & 0.01532 & 0.03678 & \
0.14488 & 0.38463 & 0.308 \\
\hline
$(15, 15, 15, 15, 15, 20, 20, 30)$ &0.00128 & 0.01532 & 0.03678 & \
0.14488 & 0.38460  & 0.748 \\
\hline
$(15, 15, 15, 15, 15, 30, 30, 60)$ & 0.00128 & 0.01532 & 0.03678 & \
0.14488 & 0.38460 & 1.712 \\
\hline
\hline
Fourier comparison & 0.00128 & 0.01532 & 0.03678 & 0.14488 & 0.38460 & 1.02 \\
\hline
\end{tabular}
\caption{Call option prices for the process with Parameter set 2, $\sigma =0$, $a=0.133896$. Option parameters $S_0=10$, $K=11$, $r=0.03$. Time is the time to compute the 5 option prices in seconds.}\label{tab:opt_calc_otm}
\end{center}
\end{table}
\end{example}
\begin{example}[\textbf{Option Theta}]
Theorems \ref{theo:at_the_money}, \ref{theo:in_the_money}, and \ref{theo:out_the_money} also lead directly to analytic expressions for the option thetas, i.e. $\frac{\partial C}{\partial T}$ and $\frac{\partial P}{\partial T}$. These can be derived simply by differentiating the formulas found in the theorems; note that all series can be differentiated termwise to yield again convergent series and that the derivatives of the functions $\{\varphi_n(T;c)\}_{n\geq 0}$ are easily derived from the derivatives of the $\{\text{Hh}_n(T;c)\}_{n\geq -1}$ functions and the recursive formula \eqref{eq:hh_deriv}. As an example, consider an ITM put option such that the underlying process has $\sigma = 0$ and $a > 0$.  According to the put-call parity and Theorem \ref{theo:in_the_money} (ii) the price is given by
\begin{align*}
P(T) = e^{-r T}K\left(\sum_{n=1}^{\infty}b_{n+1}\frac{T^{n}}{n!} + \ind(T \geq c)\sum_{n=1}^{\infty}c_{n+1}\frac{(T-c)^{n}}{n!} + 1\right) - S_0,
\end{align*}
so that for $T < c$ we have
\begin{align*}
P'(T) = e^{-r T}K\left(\sum_{n=1}^{\infty}b_{n+1}\frac{T^{n-1}}{(n-1)!} - r \sum_{n=1}^{\infty}b_{n+1}\frac{T^{n}}{n!} - r\right),
\end{align*}
and for $T > c$ we have
\begin{align*}
P'(T) = e^{-r T}K\left(\sum_{n=1}^{\infty}b_{n+1}\frac{T^{n-1}}{(n-1)!} + \sum_{n=1}^{\infty}c_{n+1}\frac{(T-c)^{n-1}}{(n-1)!} -r\sum_{n=1}^{\infty}b_{n+1}\frac{T^{n}}{n!}  -r\sum_{n=1}^{\infty}c_{n+1}\frac{(T-c)^{n}}{n!} - r \right).
\end{align*}
At $c$ the derivative will jump by the amount $e^{-rc}Kc_2 = e^{-rc}K(ak^{\eta_0/a}) = aK (K/S_0)^{(\eta_0 - r)/a}$. The put option price for an option with
parameters $S_0 = 10$, $K=11$, and $r = 0.03$, where the underlying process is defined by Parameter set 2, $\sigma = 0$ and $a=0.133896$, is plotted together with its time derivative in Figure \ref{fig:price_deriv}. At the point $c = 0.71182$ the derivative has a jump of size $0.03559$. We use the truncation $\mathbf{M} = (15,15,15,15,15,20,20,30)$ to compute both the price and the derivative.\demo
\begin{figure}[!t]
\centering
\subfloat[]
{\label{fig:put_price}\includegraphics[scale=0.48]{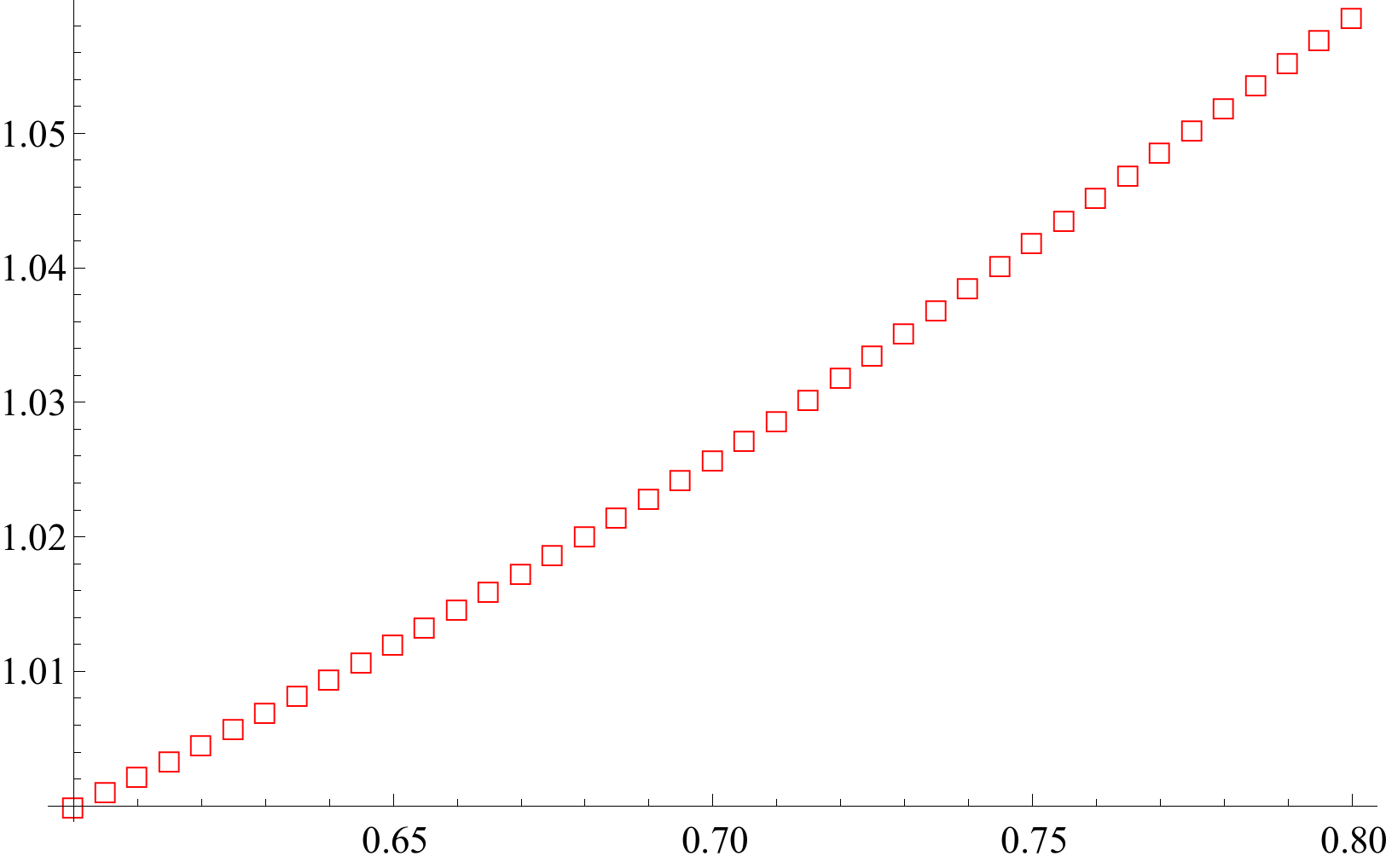}}
\subfloat[]
{\label{fig:put_price_deriv}\includegraphics[scale=0.48]{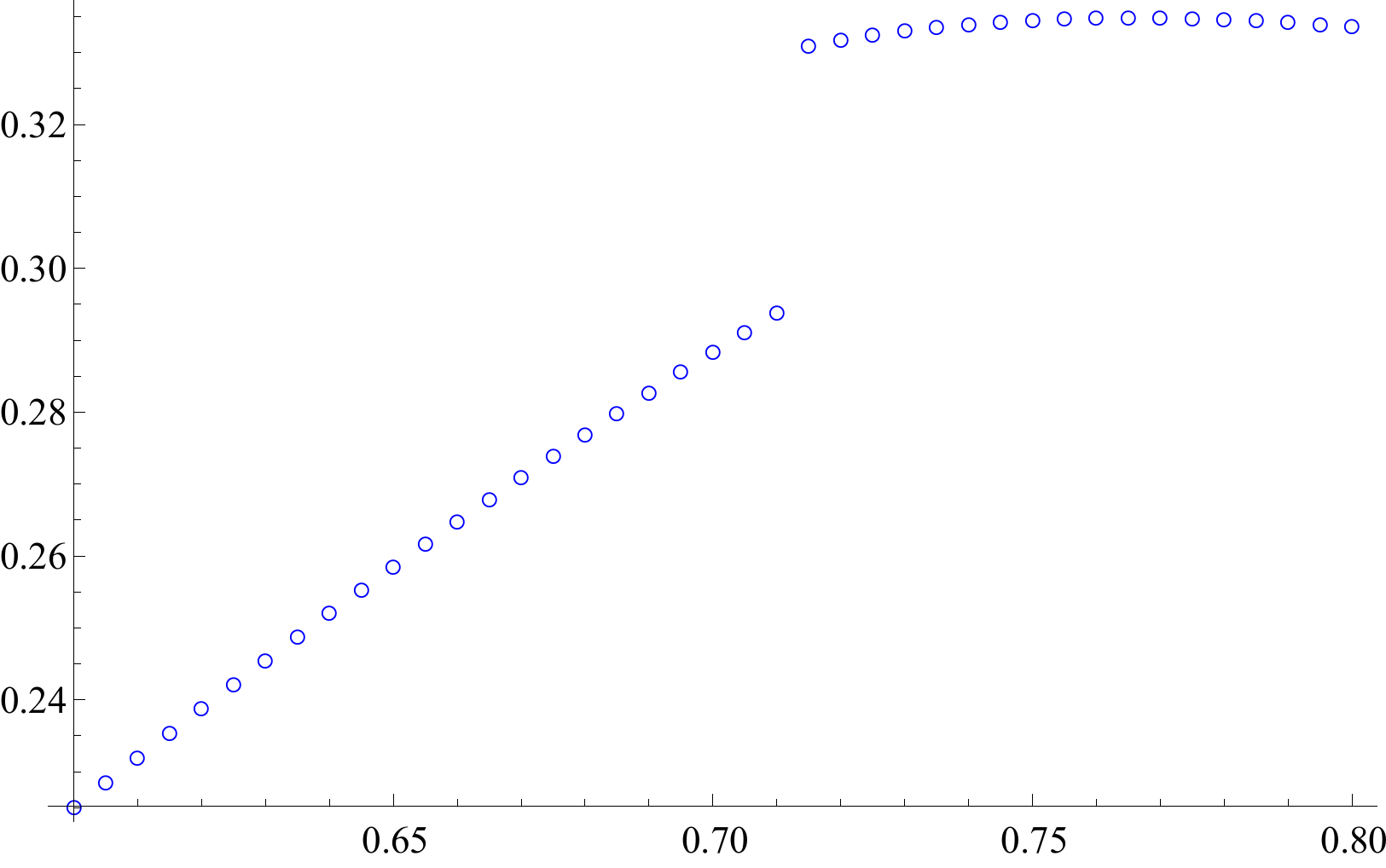}}
\caption{The put price in \ref{fig:put_price} and its derivative in \ref{fig:put_price_deriv}.}
\label{fig:price_deriv}
\end{figure}
\end{example}
\noindent To compute option deltas and gammas, i.e. the first and second derivatives with respect to $S_0$, we have to expend a little more effort and make the additional assumption that $\sigma > 0$. Going forward we will write $f(t,k)$ for the function $f(t)$ defined in \eqref{eq:def_of_f} and $p_t(x)$ for the density of $X_t$. In the Appendix in Proposition \ref{prop:cont_deriv} we show that under the assumption of risk neutrality and $\sigma > 0$ we have $\frac{\partial f(t,k)}{\partial k} = -\p(e^{X_t} > k)$, $\frac{\partial^2 f(t,k)}{\partial k^2} = \frac{p_t(\log(k))}{k}$ and that both $\frac{\partial f(t,k)}{\partial k}$ and $\frac{\partial^2 f(t,k)}{\partial k^2}$ are continuous functions of $t$ for fixed $k$. Therefore, we can carry out the same program with the Greeks as we did with the option price itself. We present this only in abbreviated form. The Laplace transform of $\frac{\partial f(t,k)}{\partial k}$ in $t$ is given by
\begin{align*}
-\int_{0}^{\infty}e^{-qt}\p(e^{X_t} > k)\d t = -\frac{\p(e^{X_{\eq}} > k)}{q} = -\begin{dcases}
\sum_{\ell = 1}^{M}\frac{\zeta'_{\ell}k^{-\zeta_{\ell}}}{\zeta_{\ell}} & k > 1, \\
\frac{1}{q} - \sum_{\ell = 1}^{\hat{M}}\frac{\hat{\zeta}'_{\ell}k^{\hat{\zeta}_{\ell}} }{\hat{\zeta}_{\ell}} & k \leq 1
\end{dcases},
\end{align*}
where the last equality is valid for $\vert q \vert$ large enough. Now with $\{\xi_{\ell}\}_{1 \leq \ell \leq M}$, $\{\hat{\xi}_{\ell}\}_{1 \leq \ell \leq \hat{M}}$ as defined in \eqref{eq:xi_varsig} and $\{\beta_{\ell}\}_{1 \leq \ell \leq M}$, $\{\hat{\beta}_{\ell}\}_{1 \leq \ell \leq \hat{M}}$  as defined in \eqref{eq:beta_1} and \eqref{eq:beta_2}  we redefine
\begin{align*}
b_n := s_n(m_{1,n},\,m_{2,n},\,\ldots,\,m_{N,n}),\quad c_n := m_{M,n},\quad m_{\ell,n} := m_{n}(\bar{\xi}_{\ell,n},\,\bar{\beta}_{\ell,n},\,\bar{\zeta'}_{\ell,n}), \quad 1 \leq \ell \leq M,
\end{align*}
and
\begin{align*}
\hat{b}_n := s_n(\hat{m}_{1,n},\,\hat{m}_{2,n},\,\ldots,\,\hat{m}_{\hat{N},n}),\quad \hat{c}_n := \hat{m}_{\hat{M},n},\quad \hat{m}_{\ell,n} := m_{n}(\bar{\hat{\xi}}_{\ell,n},\,\bar{\hat{\beta}}_{\ell,n},\,\overline{\hat{\zeta'}}_{\ell,n}), \quad 1 \leq \ell \leq \hat{M},
\end{align*}
so that
\begin{align*}
\sum_{\ell = 1}^{M}\frac{\zeta'_{\ell}k^{-\zeta_{\ell}}}{\zeta_{\ell}} = \sum_{n=2}^{\infty}\frac{b_n}{q^{n}} + \sum_{n=2}^{\infty}\frac{c_nk^{-(2q/\sigma^2)^{1/2}}}{q^{n/2}},\quad\text{and}\quad\sum_{\ell = 1}^{\hat{M}}\frac{\hat{\zeta}'_{\ell}k^{\hat{\zeta}_{\ell}}}{\hat{\zeta}_{\ell}} = \sum_{n=2}^{\infty}\frac{\hat{b}_n}{q^{n}} + \sum_{n=2}^{\infty}\frac{\hat{c}_nk^{(2q/\sigma^2)^{1/2}}}{q^{n/2}}.
\end{align*}
It follows  that
\begin{align}\label{eq:part_gre}
\frac{\partial f(t,k)}{\partial k} = 
\begin{dcases}
-\sum_{n=1}^{\infty}b_{n+1}\frac{t^n}{n!} - \sum_{n=0}^{\infty}c_{n+2}\varphi_n\left(t;c\right) & k > 1,\,c = \frac{\sqrt{2}\log(k)}{\sigma} \\
\sum_{n=1}^{\infty}\hat{b}_{n+1}\frac{t^n}{n!} + \sum_{n=0}^{\infty}\hat{c}_{n+2} \varphi_n\left(t;c\right) - 1& k \leq 1,\,c =  -\frac{\sqrt{2}\log(k)}{\sigma}
\end{dcases}
, \quad t \geq 0,
\end{align}
where the same convergence properties apply as for the series in Theorems \ref{theo:in_the_money} (i) and \ref{theo:out_the_money} (i).\\ \\
\noindent An application of Fubini's Theorem then shows that
\begin{align*}
\mathcal{L}\left\{\frac{p_t(\log(k)}{k}\right\}(q) = \frac{\d}{\d k}\left(-\frac{\p(e^{X_{\eq}} > k)}{q}\right) = \frac{1}{k}\times\begin{dcases}
\sum_{\ell = 1}^{M}\zeta'_{\ell}k^{-\zeta_{\ell}} & k > 1\\
 \sum_{\ell = 1}^{\hat{M}}\hat{\zeta}'_{\ell}k^{\hat{\zeta}_{\ell}} & k \leq 1
\end{dcases}.
\end{align*}
Redefining again
\begin{align*}
b_n := s_n(m_{1,n},\,m_{2,n},\,\ldots,\,m_{N,n}),\quad c_n := m_{M,n},\quad m_{\ell,n} := m_{n}(\bar{\beta}_{\ell,n},\,\bar{\zeta'}_{\ell,n}), \quad 1 \leq \ell \leq M,
\end{align*}
and
\begin{align*}
\hat{b}_n := s_n(\hat{m}_{1,n},\,\hat{m}_{2,n},\,\ldots,\,\hat{m}_{\hat{N},n}),\quad \hat{c}_n := \hat{m}_{\hat{M},n},\quad \hat{p}_{\ell,n} := m_{n}(\bar{\hat{\beta}}_{\ell,n},\,\overline{\hat{\zeta'}}_{\ell,n}), \quad 1 \leq \ell \leq \hat{M},
\end{align*}
gives
\begin{align*}
\sum_{\ell = 1}^{M}\zeta'_{\ell}k^{-\zeta_{\ell}} = \sum_{n=2}^{\infty}\frac{b_n}{q^{n}} + \sum_{n=1}^{\infty}\frac{c_nk^{-(2q/\sigma^2)^{1/2}}}{q^{n/2}},\quad\text{and}\quad\sum_{\ell = 1}^{\hat{M}}\hat{\zeta}'_{\ell}k^{\hat{\zeta}_{\ell}} = \sum_{n=2}^{\infty}\frac{\hat{b}_n}{q^{n}} + \sum_{n=1}^{\infty}\frac{\hat{c}_nk^{(2q/\sigma^2)^{1/2}}}{q^{n/2}}
\end{align*}
and therefore,
\begin{align}\label{eq:part_gre_2}
\frac{\partial^2 f(t,k)}{\partial k^2} = \frac{1}{k} \times
\begin{dcases}
\sum_{n=1}^{\infty}b_{n+1}\frac{t^n}{n!} + \sum_{n=-1}^{\infty}c_{n+2}\varphi_n\left(t;c\right) & k > 1,\,c = \frac{\sqrt{2}\log(k)}{\sigma} \\
\sum_{n=1}^{\infty}\hat{b}_{n+1}\frac{t^n}{n!} + \sum_{n=-1}^{\infty}\hat{c}_{n+2} \varphi_n\left(t;c\right)& k \leq 1,\,c =  -\frac{\sqrt{2}\log(k)}{\sigma}
\end{dcases}
, \quad t \geq 0,
\end{align}
where we have defined $\varphi_{-1}(t;c) := \frac{\exp(-c^2/(4t))}{\sqrt{4t}}$ and the same convergence properties apply as for the series in Theorems \ref{theo:in_the_money} (i) and \ref{theo:out_the_money} (i). Thus, writing the call/put option price now as a function also of the argument $S_0$, with $k = K/S_0$ as usual, we have:
\begin{theorem}
Suppose that $X$ fulfils the risk neutral condition and  $\sigma > 0$. Then
\begin{align*}
\frac{\partial C(T,S_0)}{\partial S_0} = e^{-rT} \left(f(T,k) - k\frac{\partial f(T,k)}{\partial k}\right), \qquad 
\frac{\partial^2 C(T,S_0)}{\partial S_0^2} = \frac{e^{-rT}}{S_0}k^2\frac{\partial ^2 f(T,k)}{\partial k ^2}
\end{align*}
where $\frac{\partial f(t,k)}{\partial k}$ and $\frac{\partial^2 f(t,k)}{\partial k^2}$  have the form \eqref{eq:part_gre} and  \eqref{eq:part_gre_2} respectively.
\end{theorem}
\begin{remark}
The put-call parity implies that $\frac{\partial P(T,S_0)}{\partial S_0} =\frac{\partial C(T,S_0)}{\partial S_0} - 1$ and $\frac{\partial^2 P(T,S_0)}{\partial S_0^2} = \frac{\partial^2 C(T,S_0)}{\partial S_0^2} $\rdemo
\end{remark}
\begin{example}[\textbf{Option Deltas and Gammas}]
We compute  $\frac{\partial C(0.1,S_0)}{\partial S_0}$ and $\frac{\partial^2 C(0.1,S_0)}{\partial S_0^2}$ for an option with parameters $K=10$, $r=0.03$, where the underlying process is defined by Parameter set 1, $\sigma = 0.042$ and $a=0.141875$. The results together with the price are shown in Figure \ref{fig:delta_gamma}. The truncations $\mathbf{M} = \hat{\mathbf{M}} = \{10, 10, 10, 10, 10, 12, 12, 16\}$ are used everywhere.\demo
\begin{figure}[!t]
\centering
\subfloat[]
{\label{fig:cp}\includegraphics[scale=0.35]{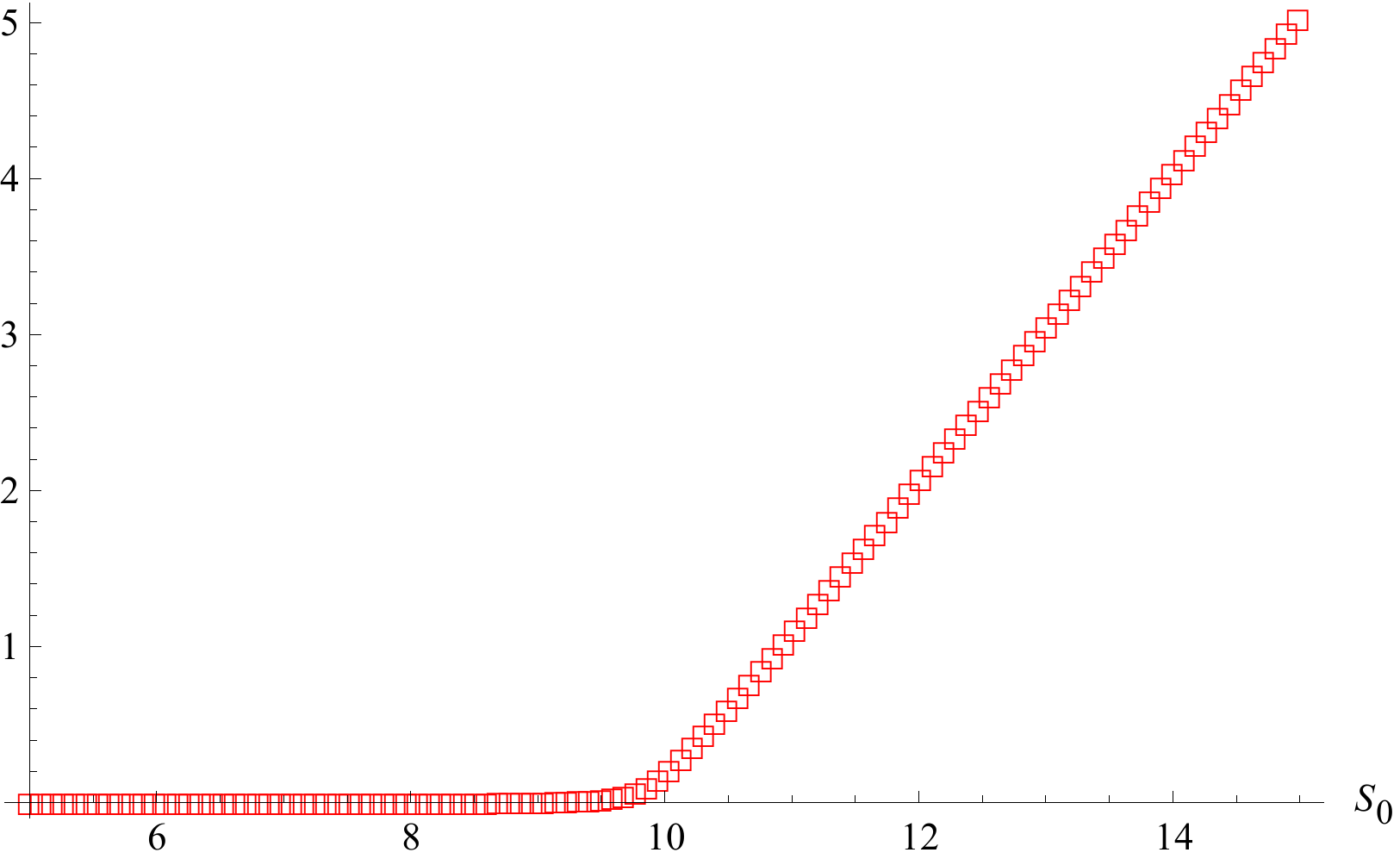}}
\subfloat[]
{\label{fig:cpd}\includegraphics[scale=0.35]{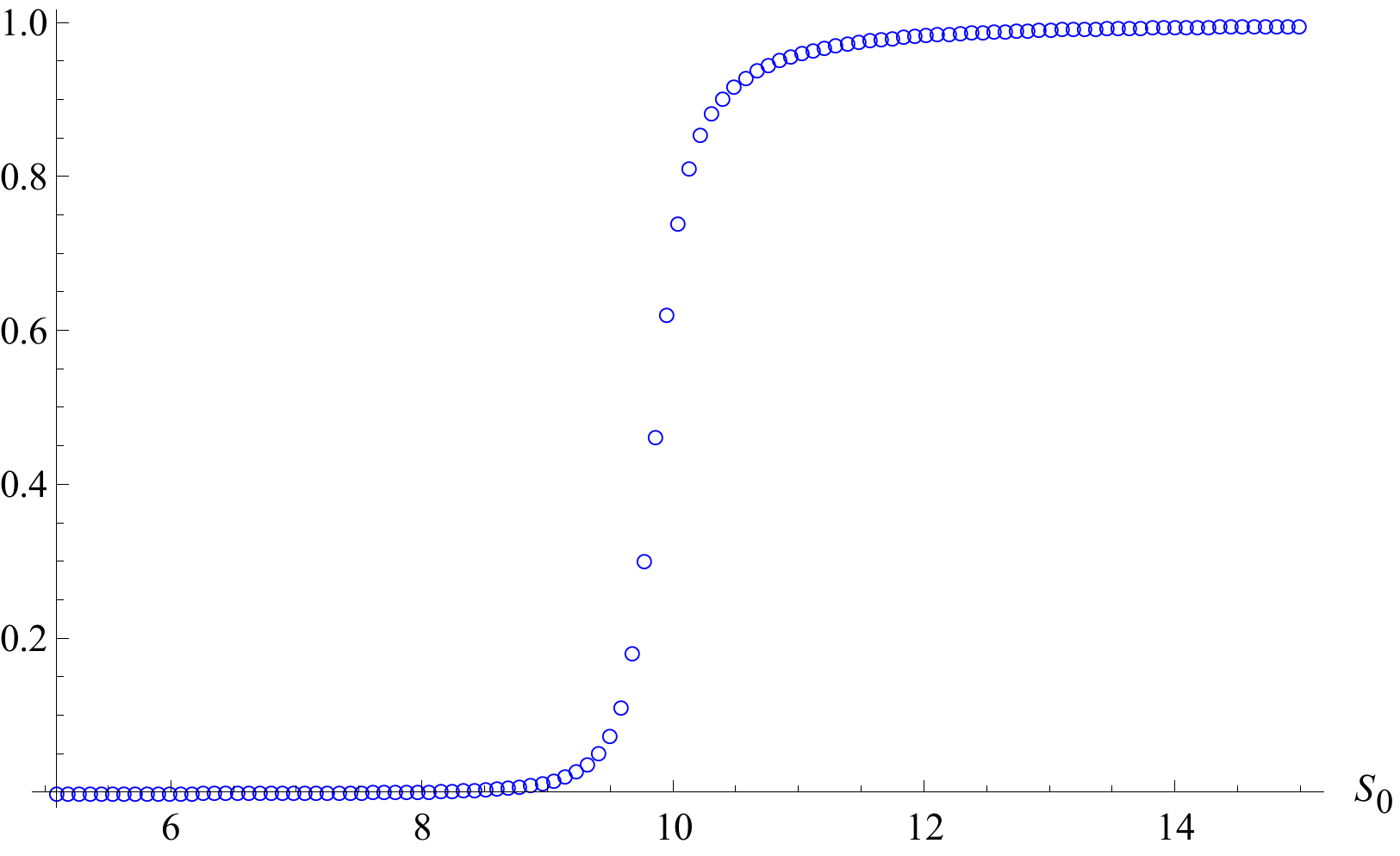}}
\subfloat[]
{\label{fig:cpd2}\includegraphics[scale=0.35]{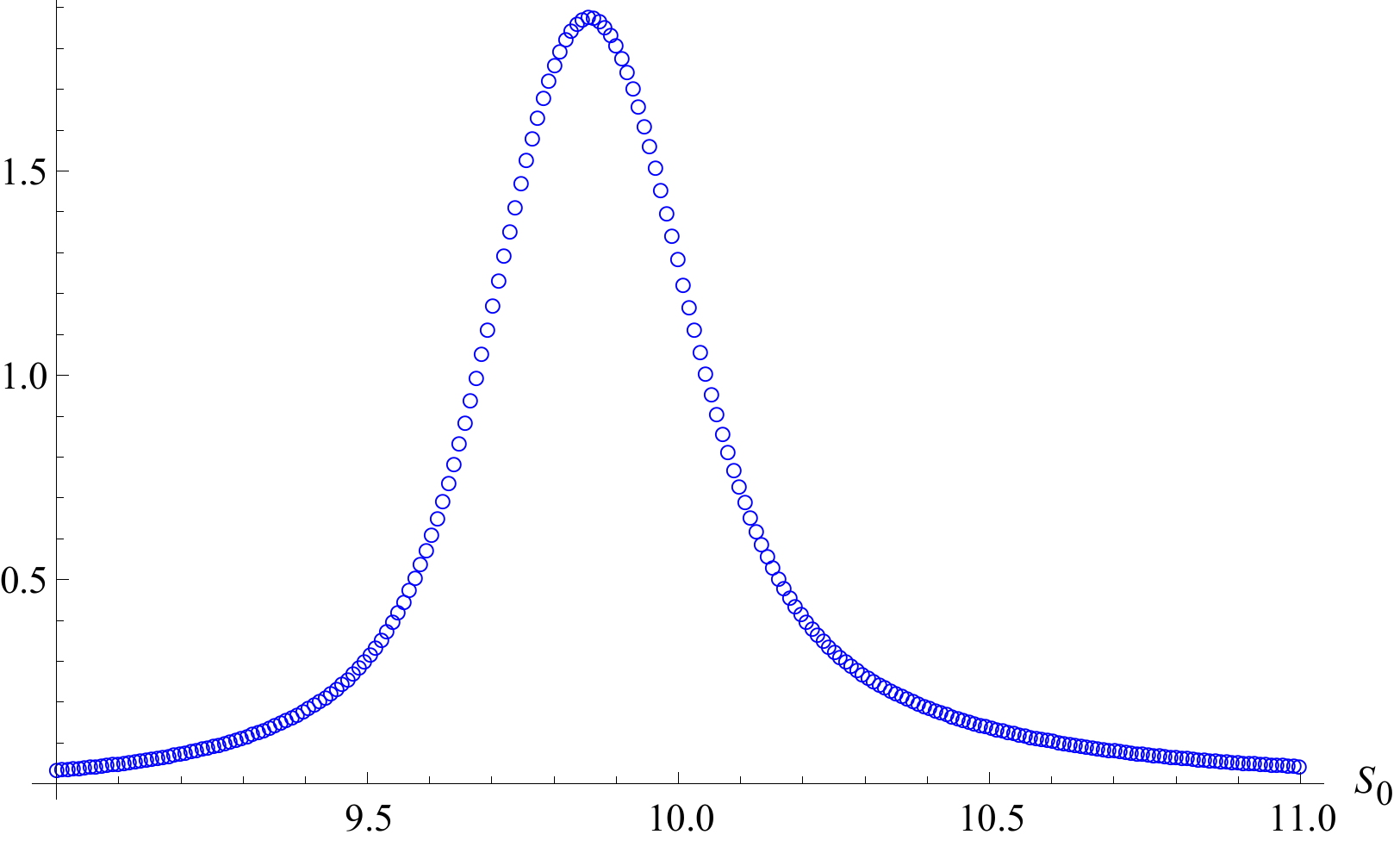}}
\caption{The price $C(0.1,S_0)$ in \ref{fig:cp}, the option delta $\frac{\partial C(0.1,S_0)}{\partial S_0}$ in \ref{fig:cpd} and the option gamma $\frac{\partial^2 C(0.1,S_0)}{\partial S_0^2}$ in \ref{fig:cpd2}}\label{fig:delta_gamma}
\end{figure}
\end{example}
\begin{remark}[\textbf{Extensions and potential extensions}]
The techniques we have outlined here are not restricted to just call and put options or financial applications. For example, it would be relatively simple to develop analytic formulas for the cumulative distribution function and density of $X_t$ (the latter provided $\sigma > 0$); we would just need to modify the approach we took for computing option deltas and gammas slightly. Additionally, looking back to Example \ref{ex:barr}, we see that the price of the up and out digital option is also expressed as a sum of basic transformations of the solutions $\{\zeta_{\ell}\}_{1\leq\ell\leq M}$. It would not be difficult (although perhaps slightly more tedious) to replicate our approach for this type of option; indeed it should be possible to adapt the approach more generally also to barrier or look-back options.\\ \\
\noindent Additionally, it is reasonable to assume that an extension to other related processes is also possible. For example the meromorphic class of processes (see e.g. \cite{KuzKyPa2011}) can be thought of as a generalization of the hyperexponential process; the key difference is that we allow $N$ and $\hat{N}$ to take the value $+\infty$ so that our Laplace exponent is no longer a rational function, but a nonetheless tractable meromorphic function. If we can justify various changes of the order of summation, most notably in Formulas \ref{eq:psi_coef}, \ref{eq:g_series}, \ref{eq:gh_series}, then the approach should generalize also to this class. Additionally, we could consider the class of processes with jumps of rational transform, essentially those L\'{e}vy processes whose Laplace exponent is a rational function (includes L\'{e}vy processes of phase-type). The complication in this case would be that we may have non-real poles with higher multiplicity and that we lose the interlacing property. 
\end{remark}\rdemo
\section*{Acknowledgement} This work was supported by the Austrian Science Fund (FWF) under the project F5508-N26, which is part of the Special Research Program ``Quasi-Monte Carlo Methods: Theory and Applications".
\newpage
\begin{appendices}
\section{Additional Proofs}\label{app:proof}
\noindent We will say a L\'{e}vy process $X$ has a $k$-th exponential moment, for $k \in \r$, if $\e[e^{kX_t}] < \infty$ for all $t \geq 0$. An equivalent statement (see Theorem 3.6 in \cite{Kyprianou}) is that $\int_{\vert x \vert > 1}e^{k}\nu(\d x) < \infty$, where $\nu(\d x)$ is the L\'{e}vy measure. We note that any hyperexponential process that satisfies the risk neutral condition has a $(1+\varepsilon)$-th exponential moment, where $0 < \varepsilon < \rho_1 - 1$.
\begin{proposition}\label{prop:cont_deriv}
If $X$ is a L\'{e}vy process with a first exponential moment and $\sigma > 0$ then
\begin{align*}
\frac{\partial f(t,k)}{\partial k} = -\p(e^{X_t} > k),\quad\text{and}\quad \frac{\partial^k f(t,k)}{\partial k^2} = \frac{p_t(\log(k))}{k},
\end{align*}
where $p_t(x)$, $x \in \r$, is the density of $X_t$.
Further, for fixed $k$, both $\frac{\partial f(t,k)}{\partial k}$ and $\frac{\partial^k f(t,k)}{\partial k^2}$ are continuous functions of $t$. 
\end{proposition}
\begin{proof}
We remark first that $X_t$ has a smooth density for all $t > 0$ since $\sigma > 0$  (see E. 29.14 in \cite{Sato}).  The first identity follows from the formula
\begin{align*}
f(t,k) = \int_{\log(k)}^{\infty}e^{x}\d F_t(x),
\end{align*}
where $F_t(x) = \p(X_t > x)$. An application of integration by parts followed by differentiation with respect to $k$ gives the result. From the stochastic continuity of $X$ it follows that $X_s$ converges in distribution to $X_t$ as $s \rightarrow t$. Since $X_t$ has a density, $(\log(k),\infty)$ is a continuity set for $X_t$, and by the Portmanteau lemma we have $\lim_{s\rightarrow t}\p(e^{X_s} > k) = \p(e^{X_t} > k)$, which proves the continuity of $\frac{\partial f(t,k)}{\partial k}$. Then, by differentiating with respect to $k$, we get immediately the formula for $\frac{\partial^k f(t,k)}{\partial k^2}$. To show continuity in $t$ is suffices to prove this for the density $p_t(x)$. It is, however, relatively easy to show that 
\begin{align*}
\lim_{s\rightarrow t}\int_{-\infty}^{\infty}\vert\phi_s(z) - \phi_t(z) \vert\d z \rightarrow 0,
\end{align*}
where $\phi_s(z) = e^{t\psi(iz)},\,z \in \r$, is the characteristic function of $X_s$. Therefore, for $x \in \r$ we have
\begin{align*}
\vert p_s(x) - p_t(x) \vert = \left \vert \frac{1}{2\pi}\int_{-\infty}^{\infty}e^{ixz}\phi_s(z)\d z - \frac{1}{2\pi}\int_{-\infty}^{\infty}e^{ixz}\phi_t(z)\d z \right \vert \leq \int_{-\infty}^{\infty}\vert\phi_s(z) - \phi_t(z) \vert\d z,
\end{align*}
and, in fact, uniform continuity follows.
\end{proof}

\end{appendices}

\end{document}